\documentclass[a4paper,11pt]{article}

\usepackage[english]{babel}
\usepackage[utf8]{inputenc} 
\usepackage{amssymb,amsmath,amsfonts,amsthm,bbm}
\usepackage{etoolbox,ifthen,xifthen}
\usepackage{hyperref,url}
\usepackage{graphicx}
\usepackage{tabularx}
\usepackage{xcolor}
\usepackage[normalem]{ulem}
\usepackage[nottoc]{tocbibind}
\usepackage{cancel}
\usepackage{authblk}
\usepackage{cite}
\usepackage{float}
\restylefloat{table}

\def\Rset{\mathbb{R}}
\def\Nset{\mathbb{N}}
\def\arsinh{\operatorname{arsinh}}
\def\esssup{\operatorname{esssup}}
\def\argmax{\operatorname*{argmax}}
\def\sign{\operatorname{sign}}
\def\D{\mathcal{D}}
\def\E{\mathcal{E}}

\def\O{\mathcal{O}}
\def\one{\mathbbm{1}}
\newcommand{\minpdf}[3]{g_{#1,#2,#3}}
\newcommand{\gtd}[3]{\mathfrak{c}_{#1,#2,#3}}
\def\support{\mathcal{X}}
\def\gauss{g}

\def\pdf{f}
\def\cdf{F}
\def\ccdf{\overline{F}}
\def\pdfcm{h}
\newcommand{\wcdf}[1]{F_{#1}}
\newcommand{\hypgeom}[2]{{\mbox{}_{#1}F_{#2}}}
\newcommand{\descort}[2][]{\mathfrak{E}_{#1}\ifthenelse{\isempty{#2}}{}{ {\left[ #2 \right]}}}
\newcommand{\escort}[2][]{\varepsilon_{#1}\ifthenelse{\isempty{#2}}{}{ {\left[ #2 \right]}}}

\definecolor{rougeG}{rgb}{.76,0,.12}
\definecolor{vertG}{rgb}{.07,.56,.25}

\newtheorem{theorem}{Theorem}
\newtheorem{corollary}{Corollary}

\newtheorem{lemma}{Lemma}
\newtheorem{definition}{Definition}

\newtheorem{property}{Property}

\newtheorem{remark}{Remark}

\usepackage{comment}
\title{Some informational inequalities involving generalized trigonometric functions and a new class of generalized moments}
\author[1]{D. Puertas-Centeno}
\affil[1]{Área de Matemática Aplicada, ESCET, Universidad Rey Juan Carlos, C/ Tulipán s/n, 28933 Móstoles, Madrid, Spain}

\author[2,3]{S. Zozor}
\affil[2]{ Univ. Grenoble Alpes, CNRS, Grenoble INP*, GIPSA-Lab, 38000 Grenoble, France}
\affil[*]{Institute of Engineering Univ. Grenoble Alpes}
\affil[3]{Instituto de Física La Plata (IFLP), CONICET-UNLP, 1900 La Plata, Argentina}

\date{\today}

\begin{document}
\maketitle

\begin{abstract}
In this work, we define a family of probability densities involving the generalized trigonometric functions defined by Drábek and Manásevich~\cite{Drabek99}, which we name Generalized Trigonometric Densities. We show their relationship with the generalized stretched Gaussians and other types of laws such as logistic, hyperbolic secant, and raised cosine probability densities. We prove that, for a fixed generalized Fisher information, this family of densities is of minimal Rényi entropy. Moreover, we introduce generalized moments via the mean of the power of a deformed cumulative distribution. The latter is defined as a cumulative of the power of the probability density function, this second parameter tuning the tail weight of the deformed cumulative distribution. These generalized moments coincide with the usual moments of a deformed probability distribution with a regularized tail.
We show that, for any bounded probability density, there exists a critical value for this second parameter below which the whole subfamily of generalized moments is finite for any positive value of the first parameter (power of the moment). In addition, we show that such generalized moments satisfy remarkable properties like order relation w.r.t. the first parameter, or adequate scaling behavior. Then we highlight that, if we constrain such a generalized moment, both the Rényi entropy and generalized Fisher information achieve respectively their maximum and minimum for the generalized trigonometric densities. Finally, we emphasis that GTDs and cumulative moments can be used to formally characterize heavy-tailed distributions, including the whole family of stretched Gaussian densities.
\end{abstract}


\section{Introduction}

Over the last decades, the design of informational tools to grasp different aspects of probability density's spreading has known a regain of active interest~\cite{Lesne14, Jizba04, Jizba04b, Tsallis22}. Rényi entropy has been proven to be related to fractal dimensions and multifractal analysis~\cite{Chhabra89, Beck93, Harte01, Chen20} and also to Kolmogorov entropy~\cite{Zmeskal13} (also called \textit{metric entropy}, see~\cite{Garrido11}). In the discrete context, Tsallis entropy~\cite{Tsallis22} (one-to-one mapped with Rényi's one) has been proven to be the only composable\footnote{An entropy function is said to be composable when the entropy of a compound system can be calculated in terms of the entropy of its independent components~\cite{Enciso17}.} entropy function with trace form\footnote{An entropy function $S$ of a probability set $p=\{p_i\}$ is said to be with trace form when there is a concave function $f$ such that $S(p)=\sum_i f(p_i)$, with the constraints $f(0)=f(1)=0$. For continuous density function, the sum is to be thought of as an integral. Equivalently, in the quantum case, the entropy of a system described by a density matrix $\hat\rho$ is with trace form when $S$ can be expressed as $S(\hat\rho)= \operatorname{Tr} f(\hat\rho)$.}, by virtue of the Enciso-Tempesta theorem~\cite{Enciso17}.
    
Far beyond their theoretical interest, Shannon and Rényi entropies have a wide range of applications~\cite[\& Refs.]{Cover06}, \cite{Bialynicki84, Bialynicki06, Zozor08}, and among them in the framework of the maximum entropy (MaxEnt) principle advocated by Jaynes in particular~\cite{Jaynes57, Mnatsakanov2008}. For instance, formal reconstruction methods of probability densities from the knowledge of all their moments of integer order are widely studied, but it is well-known that from a numerical point of view, the reconstruction problem is unstable~\cite{Talenti1987} and can even be ill-posed (see Hamburger problem~\cite{Hamburger1920:I, Hamburger1920:II,  Hamburger1921, Stieltjes1884,  Stieltjes1886, Hausdorff1921:I,  Hausdorff1921:II,  Shohat1970, Lasserre2010}). This fact leads some authors to tackle the problem through the lens of the MaxEnt principle~\cite{Biswas2010}, which allows to approximate reconstruction of the probability density from a finite number of fixed moments $\langle |x|^p\rangle$ and maximizing the entropy of the law subject to the moment constraints. The MaxEnt principle is also sometimes used in Bayesian inference to construct a prior as non-informative as possible from fixed moments~\cite{Robert2007}.

While Shannon entropy is used for extensive systems description, Rényi (or equivalently Tsallis) entropy is used in the non-extensive context~\cite{Tsallis22}. However, when dealing with probability densities with heavy-tailed (e.g., modeling rare or sudden events) only a few moments are finite which limits the recourse to the MaxEnt approach to model the law of data. To solve this drawback, the use of fractional moments~\cite{Gzyl2010_Stieltjes, Novi2003} or escort mean values~\cite{Tsallis09} has been investigated. In this work, we define a set of density functionals we will call \textit{cumulative moments}, which provides an infinite and well-defined sequence for any heavy-tailed density with tails decaying as $x^{-\eta}, \: \eta>1$. Moreover, we show that this set allows to reconstruct (under certain assumptions) the probability density in a formal way. Furthermore, we can express the probability density that maximizes the Shannon or Rényi entropies subject to the constraint of fixed cumulative moments in terms of the probability density that maximizes the corresponding entropic functional with the constraint of fixed standard moments.

As for the usual moments, we prove that the cumulative moments are bounded by Shannon and Rényi entropies as well as by the inverse of a generalized Fisher information which extends the relations of the same kind involving usual moments, Shannon or Rényi entropies, and usual Fisher information.

All along the paper, our results heavily lie on the so-called \textit{differential-escort} transformations, which 
consist of a power-like stretching, but keep invariant the probability in each differential interval of the support through a nonlinear change of variables. They are very similar to the combined Sundman transformations introduced a long time ago in~\cite{Sundman1913}, and more recently used in bubble dynamics~\cite{Kudryashov14, Kudryashov15}, or in the study of the differential equations involved in ultrasonic waves~\cite{Gordoa21}. Incidentally, such a transform has the ability to modify the tail of the distribution it is applied to. For example, the differential-escort transformations of an exponential density turn to be $q-$exponential density, maximizing distribution of the Tsallis entropy under first-order moment constraint~\cite{Tsallis88}, this distribution exhibiting a power law decay. Conversely, as we will see in the sequel, any power-law density with connected support can be reversibly transformed into an exponential decaying density (with connected support), and even into a connected compact support density.

A set of mathematical functions playing a central role in this paper is the so-called family of generalized trigonometric functions (GTFs). As far as we know, these functions appear for the first time in literature in Shelupsky's work~\cite{Shelupsky59, Burgoyne64}, but have begun to gain broader interest only since Lindqvist's~\cite{Lindqvist95} and Drábek and Manásevich papers~\cite{Drabek99}. In the last decades, GTFs have been the subject of active research, not only in terms of its mathematical properties~\cite{Lindqvist95, Drabek99, Bhayo12, Takeuchi12, Baricz14, Yin19, Takeuchi21, Miyakawa21, Wang24} but also for their physical applications. For instance, the versatility of GTFs together with the fact they are solutions of an eigenvalue $p-$Laplacian problem has allowed that these functions have been used to study nonlinear spring-mass systems in Lagrangian Mechanics~\cite{Wei12}, to describe vibration properties in some solids~\cite{Cveticanin20a, Vujkov22}, to describe nonlinear oscillators~\cite{Cveticanin20b},  to provide the generalized stiffness and mass coefficients for power-law Euler–Bernoulli beams~\cite{Skrzypacz20}, in quantum gravity~\cite{Shababi16}, in quantum graphs~\cite{Ahrami24}, and very recently to express some exact solutions of a non-linear Schrödinger equation~\cite{Gordoa24}. In this paper, we will see that the family of densities that saturate the above-mentioned inequalities is expressed via the GTFs.

Summarizing, this paper has several purposes. First, we define a family of probability densities named Generalized Trigonometric Densities (GTDs) based on GTFs~\cite{Lindqvist95, Drabek99, Yin19}. These distributions are related to the generalized Gaussian distributions precisely by the differential-escort transformation we will recall. We then highlight that, for a fixed generalized Fisher information, Rényi entropy is minimized for such densities~\cite{Zozor17}.
Secondly, we define an extended family of generalized moments introducing a second parameter tuning the tails of the distribution on which the mathematical mean is taken, in such a way that the weight of the tail can be absolutely regulated in virtue of the properties of the differential-escort transformations previously emphasized. Furthermore, we will show that this extension of the generalized moments upper-bounds the Shannon and Rényi entropies and the inverse of the generalized Fisher information~\cite{Lutwak05}. As we will see, the bound is reached again by the GTDs in a great amount of cases.

The structure of the paper is as follows. Section~\ref{Sec_Preliminary} gives some preliminary notions of generalized trigonometric functions, informational functionals, and some inequalities they satisfy, and on \textit{differential-escort} transformations. In Section~\ref{Sec_CumMom}, we propose a definition of what we name \textit{generalized cumulative moments}, and we study some basic properties they satisfy. In Section~\ref{Sec_densities}, we present the family of generalized trigonometric densities induced by the generalized trigonometric functions, we show their relationship with the generalized Gaussians distributions, and we study the behavior of their generalized moments. In Section~\ref{Sec_CumIneq}, we derive Cramér-Rao and entropy-momentum-like inequalities involving the proposed extension of the generalized moments, Rényi and Shannon entropies, and generalized Fisher information. The family of densities that saturate the inequalities and the minimal bounds are analytically obtained. Finally, in Section~\ref{Sec_Conclusions}, we draw some conclusions and discuss some open problems. 

Note that all along the paper we will place on the framework of univariate probability densities defined as densities w.r.t. the Lebesgue measure, the case for more general reference measure, and/or multidimensional probability density functions remaining as future research directions.


\section{Brief recalls on generalized trigonometric functions, information measures and stretched differential-escort transforms}\label{Sec_Preliminary}


\subsection{Generalized trigonometric functions}\label{ssec:Trig}

The so-called \textit{generalized trigonometric} and \textit{generalized hyperbolic functions} are families of functions ruled by two parameters, encompassing the trigonometric and hyperbolic functions, respectively, with which they share relevant characteristics. These families also include some relevant functions in mathematics and physics such as lemniscate function~\cite{Nishimura15, Takeuchi16, Yin22, Zhang23} and some generalizations studied by Ramanujan~\cite{Yin22}, Dixon elliptic functions~\cite{Wei12}, Ateb functions~\cite{Cveticanin20b} among others.

Let us first  give the definition of the generalized $(v,w)-$arcsine function~\cite{Drabek99},
\begin{equation} \label{def:arcsinpq}
\arcsin_{v,w} (y) = \int_0^y \left( 1 - |t|^w\right)^{-\frac1v} dt \, = \, y \: \hypgeom{2}{1} \left(\frac1v \, , \, \frac1w \, ; \, 1 + \frac1w \, ; \, |y|^w \right), \quad y \in (-1,1),
\end{equation}
were $\hypgeom{2}{1}$ denotes the Gauss hypergeometric function, and where, the usual arcsine function is recovered for $v = w = 2$. In the literature, the range of the parameters is often limited to $v, w>1$ ( e.g., in the context of an eigenvalue problem dealing with the $p-$Laplacian operator~\cite{Lindqvist95, Drabek99, Takeuchi12}). However, a simple look at the integral shows that the range of parameters for which the function is well defined (the integral converges) is not restricted to $v, w > 1$. For instance  the $(v,w)-$arcsine function was studied  for $v = w > 0$ in~\cite{Baricz14,Karp15}, for $v>w/(w+1), \; w>1$ in~\cite{Miyakawa21} and for $v>w/(w+1), \; w>0$ in~\cite{Miyakawa22}. In this paper, we explore this function and its related functions for all parameters the integral exists, i.e., $v \in \Rset^* = \Rset\setminus \{0\}$ and $w > 0$. As for the sine function, the generalized $(v,w)-$sine function is then the inverse function of the $(v,w)-$ arcsine function, ``correctly'' periodized. To introduce the  domain of the generalized $(v,w)-$sine function in a compact form, it is useful to adopt the following notation~\cite{Miyakawa21, Miyakawa22},
\begin{equation} \label{def:pipq}
\frac{\pi_{v,w}}2 = \lim_{y \to 1} \arcsin_{v,w}(y) =  \lim_{y \to 1} \int_0^{y} \left( 1 - t^w \right)^{-\frac1v} dt = \left\{
\begin{array}{cll} \displaystyle \frac1w \, B\left(1 -\frac1v \, , \, \frac1w \right) & \mbox{if} &  v\in \Rset \setminus [0,1],\\[5mm]
\infty & \mbox{if} & v \in (0,1],
\end{array}
\right.
\end{equation}
where $B(\cdot\,,\cdot)$ denotes the usual Beta function. Then, the $(v,w)-$generalized sine function $\sin_{v,w}$ is primarily defined on $\left( - \frac{\pi_{v,w}}2 , \frac{\pi_{v,w}}2 \right) \longrightarrow [-1,1]$ as the inverse function of $\arcsin_{v,w}$ given in Eq.~\eqref{def:arcsinpq}. Finally, when $\pi_{v,w}$ is finite ($v  \notin [0 , 1]$), as for the sine function its generalized version is extended to $\left( - \frac{\pi_{v,w}}2 , \frac{3 \pi_{v,w}}2 \right)$ imposing the symmetry $\sin_{v,w}(\pi_{v,w}-x) = \sin_{v,w}(x)$, and to whole $\Rset$ imposing $\pi_{v,w}-$periodicity. Let us mention that, when $v \in (0,1]$, i.e., $\pi_{v,w} = + \infty$, $\sin_{v,w}(x)$ is increasing  with a behaviour similar to the hyperbolic tangent $\operatorname{tanh}$ function~\cite{Miyakawa21}.

The $(v,w)-$generalized cosine function is defined like for the usual case as the derivative of the $(v,w)-$generalized sine function,
\begin{equation}\label{def:cospq}
\cos_{v,w}(x) = \frac{d}{dx}\sin_{v,w}(x).
\end{equation}
Immediately, as noted for the $(v,w)-$arcsine function, the usual sine and cosine functions are recovered for $v = w = 2$. Similar to the usual case, it appears that these functions satisfy an identity of Pythagorean type\footnote{Basically, one uses the inverse function rule applied to $\cos_{v,w}$ using the fact that $ \sin_{v,w}  = \arcsin_{v,w}^{-1}$, and the symmetries of the functions to conclude over $\Rset$.}:
\begin{equation}\label{eq:pyth_rel}
\left| \cos_{v,w}(x) \right|^v + \left|\sin_{v,w}(x) \right|^w = 1.
\end{equation}

On the other hand, following the Miyakawa et al. notation~\cite{Miyakawa21, Miyakawa22}, the generalized $(v,w)-$hyperbolic sine function, $\sinh_{v,w}: \left( -\frac{\pi_{r,w}}2 , \frac{\pi_{r,w}}2 \right) \longrightarrow \Rset$ is defined as the inverse function of 
\begin{equation}\label{def:arcsinhpq}
\arsinh_{v,w} (y) = \int_0^y \left( 1 + |t|^w \right)^{-\frac1v} dt \, = \, y \: \hypgeom{2}{1}\left( \frac1v \, , \, \frac1w \, ; \, 1+\frac 1w \, ; \, -|y|^w \right), \quad y \in \Rset
\end{equation}
where $\pi_{r,w}$ is defined by Eq.~\eqref{def:pipq}, with $r$ given by
\begin{equation}\label{eq:duality_param}
\frac1v + \frac1r = 1 + \frac 1w.
\end{equation}
As for the generalized cosine, the $(v,w)-$hyperbolic cosine function is  defined  as the derivative  of the $(v,w)-$hyperbolic sine function,
\[
\cosh_{v,w}(x) = \frac d{dx} \sinh_{v,w}(x)
\]
In the same line as their trigonometric counterpart, they satisfy the identity
\begin{equation}\label{eq:pyth_rel_hyperbolic}
\left| \cosh_{v,w}(x) \right|^v - \left|\sinh_{v,w}(x) \right|^w = 1.
\end{equation}
Note that $\sinh_{v,w}(x)$ is a strictly increasing function, so that, $\cosh_{v,w}(x)$ is always positive, and thus the absolute value for the first term of the identity can be skipped. Again, the usual hyperbolic functions and related identity are recovered for $v = w = 2$.

Finally, it is important to mention the duality relations recently obtained by Miyakawa et al. in~\cite{Miyakawa21, Miyakawa22} for  $w>1$ and $v>w/(w+1)$, which can be extended, without loss of generality, to any $v\in\Rset^*, \; w>1:$

\begin{equation}\label{eq:duality}
\sinh_{v,w}(x) = \frac{\sin_{r,w}(x)}{\left[ \cos_{r,w}(x)\right]^{\frac{r}{w}}},
\qquad \sin_{v,w}(x) = \frac{\sinh_{r,w}(x)}{\left[ \cosh_{r,w}(x)\right]^{\frac{r}{w}}}
\end{equation}
and
\begin{equation}\label{eq:duality2}
\cos_{v,w}(x) = \frac1{\left[ \cosh_{r,w}(x) \right]^{\frac{r}{v}}},
\qquad \cosh_{v,w}(x) = \frac1{\left[ \cos_{r,w}(x) \right]^{\frac{r}{v}}}
\end{equation}
where $r$ is given by Eq.~\eqref{eq:duality_param}.  The last duality relation, Eq.~\eqref{eq:duality2}, involving the generalized cosine and its hyperbolic counterpart, is relevant to interpret the results we will present in the sequel. 

For completeness,  let us mention the nonlinear differential equations satisfied by the generalized sine and hyperbolic sine functions, which are closely related to what is known as the initial value problem of $p-$Laplacian~\cite{Drabek99}. They can be directly derived by the Pythagorean like relations~\eqref{eq:pyth_rel} and~\eqref{eq:pyth_rel_hyperbolic}, respectively,
\begin{equation*}\left\{\begin{array}{l}
\displaystyle\Big( |u'|^{v-2} u' \Big)' \, + \, \frac{(v-1) \, w}v \, |u|^{w-2} \, u = 0\\[2.5mm]
\displaystyle \Big( |u'|^{v-2} u' \Big)' \, - \, \frac{(v-1) \, w}v \, |u|^{w-2} \, u = 0
\end{array}\right., \qquad u(0) = 0, \quad u'(0) = 1.
\end{equation*}
It is easy to check that  $u_+ = \sin_{v,w}$ is a solution of the first equation, while  $u_- = \sinh_{v,w}$ is a solution of the second equation.


\subsection{Some informational measures and inequalities}\label{ssec:inform_basic}

When thinking about ``information'', several points of view can be envisaged. Although C. Shannon did not like this word in such a context, in the so-called information theory field, the notion of information is generally understood through the prism of entropy (or more precisely, entropy is interpreted as a measure of lack of information). However, the word ``information'' also appears in the field of estimation, more precisely through the Cramér-Rao inequality where the variance of any estimator is lower-bounded by the inverse of the Fisher information. In a way, the Fisher information gives the information on a parameter contained in the data. In the sequel, we will indeed focus on the so-called nonparametric Fisher information and nonparametric Cramér-Rao inequality, so that roughly speaking, the Fisher information is related to location, or said differently, to local information or fluctuations on a probability density function. Finally, one can simply think about moments of a random variable that also characterize in a sense some information on the variable (e.g., the location via the statistical mean, the dispersion via the variance, the asymmetry via the third order moment, the impulsivity via the kurtosis/fourth order moment,\ldots). In the following paragraph, we briefly recall these three possible notions of information (entropy, Fisher information, moment) together with some relationships they satisfy.


\paragraph{Rényi and Tsallis entropies} As just introduced, the first family of informational measures is of entropic type. More precisely, we concentrate here on the Rényi and Tsallis entropies of $\lambda-$order of a probability density function $\pdf$ defined on $\Rset$ or on a Lebesgue measurable subset of $\Rset$, which are respectively defined as
\begin{equation*}
R_\lambda[\pdf] = \frac1{1-\lambda} \log\left( \int_\Rset [\pdf(x)]^\lambda \, dx \right), \qquad T_\lambda[\pdf] = \frac1{\lambda-1} \left( 1 - \int_\Rset [\pdf(x)]^\lambda \, dx \right),
\end{equation*}
where the case $\lambda = 1$ is to be thought of as the limit of the expression. The integral term can be written as a mathematical (or statistical) mean, $\left\langle \pdf^{\lambda-1}(x) \right\rangle_\pdf$, where the subscript $\pdf$ means that the mean is taken over the statistical distribution $\pdf$. Remind that, contrary to the discrete case, differential Rényi and Tsallis entropies are not necessarily positive quantities. In the sequel, some inequalities will involve the quantity called power Rényi entropy, defined as
$$N_\lambda[\pdf] = e^{R_{\lambda}[\pdf]} = \left\langle\pdf^{\lambda-1}(x)\right\rangle^\frac1{1-\lambda}_\pdf.$$
Precisely, the Shannon entropy
\begin{equation*}
S[\pdf] = - \int_\Rset \pdf(x) \, \log\pdf(x) \, dx.
\end{equation*}
is recovered in the limit case  $\displaystyle \lim_{\lambda \to 1}R_{\lambda}[\pdf]=\lim_{\lambda \to 1} T_{\lambda}[\pdf]=S[\pdf]$. We will simply denote $N[\pdf] = e^{S[\pdf]}$  the power Shannon entropy. Note that Rényi and Tsallis entropies are one-to-one mapped, $T_{\lambda}[\pdf]=\frac{e^{(1-\lambda) R_{\lambda}[\pdf]} - 1}{1 - \lambda}.$ Because the function $\displaystyle r \mapsto \frac{e^{(1-\lambda) \, r} - 1}{1 - \lambda}$ is strictly increasing, given a set of probability densities $\D_c$,  e.g., satisfying some constraints (Fisher information, moments,\ldots), necessarily Rényi and Tsallis entropies are maximized for the same probability density on $\D_c, \: \argmax_{\D_c} R_{\lambda}[\pdf] = \argmax_{\D_c} T_{\lambda}[\pdf]$.


\paragraph{$(p,\lambda)-$Fisher information.} The second class of informational measures comes from the statistical field, namely the Fisher information involved in the performance bounds of estimators (parametric context) or the fluctuations of a random variable (nonparametric context).  An extension of the (nonparametric) Fisher information tunes by two parameters was introduced by Lutwak and Bercher~\cite{Lutwak05, Bercher12, Bercher12a} as follows: Given $p>1$ and $\lambda \in \Rset^*$,  the $(p,\lambda)-$Fisher information is defined as
\begin{equation*}
F_{p,\lambda}[\pdf]
\: = \: \int_\Rset \left|\pdf(x)^{\lambda-2} \, \frac{d\pdf}{dx}(x)\right|^{p} \pdf(x) \, dx
\: = \: \frac1{|\lambda-1|^{p}} \, \left\langle \left|\frac {d \pdf^{\lambda-1}}{dx}(x)\right|^p \right\rangle_{\!\pdf},
\end{equation*}
when $\pdf$ is differentiable in the closure of his support. Throughout the paper we adopt the notation $\phi_{p,\lambda}[\pdf]=F_{p,\lambda}[\pdf]^\frac1{p\lambda}$.

Note that the $(2,1)-$Fisher information reduces to the standard Fisher information of a probability density. The $(1,\lambda)-$Fisher information corresponds to  the total variation of $\frac{\pdf^\lambda}{\lambda}$ provided that $\pdf^\lambda$ has bounded variations. When $p \to \infty$, the limit  $\displaystyle \lim_{p \to +\infty} \phi_{p,\lambda}^\lambda$ is the essential supremum of $\left|\left(\frac{\pdf^{\lambda-1}}{\lambda-1}\right)'\right|$ on the support of $\pdf$, provided $\pdf$ is absolutely continuous, and the essential supremum is finite.


\paragraph{$p-$th moment and $p-$th absolute moment.} The third wide class of informational measures is given by moments of a probability distribution. Among them, for natural $p$, the (natural) moments of a random variable with probability density function $\pdf$ are defined by
\begin{equation*}
\widetilde{\mu}_p[\pdf] = \int_\Rset x^p \, \pdf(x) \, dx  = \left\langle  \, x^p \, \right\rangle_\pdf.
\end{equation*}
Also of importance the $p-$th absolute moments for $p \geqslant 0$, are defined as 
\begin{equation*}
\mu_p[\pdf] = \int_\Rset |x|^p \, \pdf(x) \, dx  = \left\langle  \, |x|^p \, \right\rangle_\pdf.
\end{equation*}
when the integral converges. The $p-$typical deviation is then defined as the moment to the power $1/p$ (and the limit value when $p \to 0$)
\begin{equation*}
\sigma_p[\pdf] = \displaystyle \ \mbox{for} \ p > 0, \qquad \sigma_0[\pdf] = \lim_{p \to 0} \sigma_p[\pdf]  = \exp\left(\int_\Rset \pdf(x) \, \log|x| \, dx \right)
\end{equation*}
In the limit case $p \to \infty$, the limit of $\sigma_p[\pdf]$ is the essential supremum of the support of $\pdf$, i.e., $\esssup \big\{ \, |x| : \pdf(x) > 0 \, \big\}$. Note that, when it exists, the variance, as well as, the standard deviation are trivially recovered taken $p=2$ for a centered probability density\footnote{Dealing with such moments in terms of ``centrality'', i.e., around a location parameter, is not an obvious task. We postpone such a discussion in the appendices.}.


\paragraph{Inequalities involving these informational quantities}
A first inequality involving informational measures make interplaying the generalized Fisher information and the Rényi entropy power. Namely, the Stam inequality and their generalizations with two parameters establish that the product of Rényi entropy power and the $(p,\lambda)-$Fisher information is lower bounded: Given\footnote{For $p \to \infty$ the inequality also holds with $\displaystyle \lim_{p \to +\infty} \phi_{p,\lambda}^\lambda$ being the essential supremum of $\left|\left(\frac{\pdf^{\lambda-1}}{\lambda-1}\right)'\right|$ as discussed before~\cite{Lutwak04}.} $p \in [ 1 , +\infty )$ and $\lambda > \frac1{1+p^*} $~\cite{Lutwak05, Bercher12a, Zozor17},
\begin{equation}\label{ineq:bip_Stam}
\phi_{p,\lambda}[\pdf] \, N_{\lambda}[\pdf] \: \geqslant \: \phi_{p,\lambda}[\gauss_{p,\lambda}] \, N_{\lambda}[\gauss_{p,\lambda}],
\end{equation}
for any continuously differentiable probability density $\pdf$, where the minimizing densities of the product, denoted  $\gauss_{p,\lambda}$, are stretched deformed Gaussian\footnote{This distribution was named under this terminology in~\cite{Zozor17}, but is also known as generalized Gaussians~\cite{Lutwak05} or $q-$Gaussian~\cite{Bercher12a}.} laws given by 
\begin{equation}\label{def:g_plambda}
\gauss_{p,\lambda}(x) \, = \, \frac{a_{p,\lambda}}{\exp_\lambda\left( |x|^{p^*} \right)}
\, = \, a_{p,\lambda}\, \exp_{2-\lambda}\left( - |x|^{p^*} \right)
\end{equation}
with  $p > 1$, $p^* = \frac{p}{p-1}$ the Hölder conjugated of $p$, and where $\exp_\lambda$ denotes the generalized Tsallis exponential\footnote{In the case $p \to 1$,  $p^* \to \infty$, and $\gauss_{1,\lambda}$ turns to be a constant density over a unit length support. In addition, the definition of the stretched Gaussian (and the inequality) can be extended to $p^* \in (0,1)$ (i.e. $p \in (-\infty,0)$). Finally, when $p^* = p = 0$ and $\lambda > 1$,   $\gauss_{0,\lambda} = a_{0,\lambda} (-\log|x|)_+^{\frac1{\lambda-1}}$, with $a_{0,\lambda} = \frac1{2\Gamma\left(\frac{\lambda}{\lambda-1}\right)}.$}
\begin{equation}\label{def:q-exp}
\exp_\lambda(x) = \left( 1 + (1 - \lambda) \, x \right)_+^\frac1{1-\lambda}, \ \ \lambda \ne 1, \qquad \exp_1(x) \: \equiv \: \lim_{\lambda \to 1} \, \exp_\lambda(x) \: = \: \exp(x)
\end{equation}
with the notation $(\cdot)_+ = \max(0,\cdot)$, this stretched deformed Gaussian being defined for\footnote{The restriction on $\lambda$ in the inequality~\eqref{ineq:bip_Stam} is imposed by the non-existence of the right handside if not satisfied.} $\lambda  > 1-p^*$ . When $\lambda >1$ the support of $\gauss_{p,\lambda}$ is  $\support_{p,\lambda} = \left( - (\lambda-1)^{\frac1{p^*}} \: , \: (\lambda-1)^{\frac1{p^*}}\right),$ and $\support_{p,\lambda} = \Rset$ when $\lambda \leqslant 1$.
The normalization constant $a_{p,\lambda}$ is given by
\begin{equation}\label{def_aplambda}
a_{p,\lambda} =  \left\{ \begin{array}{lll}
\displaystyle  \frac{p^*\,|1-\lambda|^\frac1{p^*}}{2\,B\left(\frac1{p^*},\frac\lambda{|1-\lambda|}+\frac{\one_{\Rset_+}(1-\lambda)}{p}\right)},& \text{for}  & \lambda\neq1,\\[7.5mm]
\displaystyle \frac{p^*}{2\Gamma(\frac1{p^*})},&\text{for} &  \lambda=1,
\end{array}\right.
\end{equation}
where $B(\cdot,\cdot)$ denotes the Beta function, and $\one_A$ the indicator function, $\one_A(x) = 1$ if $x \in A$, and $0$ otherwise, and $\Rset_+ \equiv [ 0 , +\infty)$.

A second inequality involves the Rényi entropy power and moments. More specifically, when $p^* \in [0 , \infty )$ and $\lambda>\frac1{1+p^*},$ a  generalized entropy-momentum inequality was derived by Lutwak or Bercher~\cite{Lutwak04, Lutwak05, Bercher12}, which expresses as
\begin{equation}\label{ineq:bip_E-M}
\frac{\sigma_{p^*}[\pdf]}{N_{\lambda}[\pdf]} \, \geqslant \, \frac{\sigma_ {p^*}[\gauss_{p,\lambda}]}{N_{\lambda}[\gauss_{p,\lambda}]}.
\end{equation}
where the minimizers are the same as that of the generalized bi-parametric Stam inequality eq.~\eqref{ineq:bip_Stam}.

Finally, combining  Eqs.~\eqref{ineq:bip_Stam} and~\eqref{ineq:bip_E-M}, one straightforwardly obtains a bi-parametric generalization of the Cramér-Rao inequality, involving now the Fisher information and moments: Given\footnote{In the limit $p \to + \infty$, $p^* \to 1$, the inequality is still valid~\cite{Lutwak04}.} $p \in [1 , +\infty )$ and $\lambda>\frac1{1+p^*},$ 
\begin{equation}\label{ineq:bip_CR}
\phi_{p,\lambda}[\pdf] \, \sigma_ {p^*}[\pdf] \: \geqslant \: \phi_{p,\lambda}[\gauss_{p,\lambda}] \, \sigma_ {p^*}[\gauss_{p,\lambda}].
\end{equation}
Note that for $p=2$ and $\lambda=1$ the usual Cramér-Rao inequality is recovered~\cite{Rao45, Hodges51, Cover06}. 

The inequalities~\eqref{ineq:bip_Stam},~\eqref{ineq:bip_E-M} and~\eqref{ineq:bip_CR} were also extended for probability densities on $\Rset^n$~\cite{Bercher12, Lutwak05}, which goes beyond the scope of the paper. Finally, as we previously mentioned, the generalized Stam inequality Eq.~\eqref{ineq:bip_Stam} was extended in the monodimensional case by decoupling the parameter $\lambda$ of the Rényi entropy power from the second parameter of the $(p,\lambda)-$Fisher information, that is to say, with three free parameters. More precisely, for $p \geqslant 1, \; \beta > 0, \; \lambda > 1 - \beta p^*,$ the following inequality was obtained~\cite{Zozor17}:
\begin{equation}\label{ineq:trip_Stam}
\phi_{p,\beta}[\pdf] \, N_{\lambda}[\pdf] \: \geqslant \: \phi_{p,\beta}[ \minpdf {p}{\beta}{\lambda}] \, N_{\lambda}[ \minpdf {p}{\beta}{\lambda}],
\end{equation}
where the minimizers $ \minpdf{p}{\beta}{\lambda}$ are given by  special functions. We will come back to the latter, given precisely through a transform (stretched differential escort transformation) of generalized  Gaussian densities, which can be expressed through the generalized cosine functions discussed above. It is shown that the density $\gauss_{p,\lambda}$  is naturally recovered in the case $\beta = \lambda$, which trivially follows merely by comparing Eqs.~\eqref{ineq:bip_Stam}-\eqref{ineq:trip_Stam}. We will come back to these functions in more detail in Section~\ref{Sec_densities}, paying special attention to the case  $\lambda \neq 1$ and $\beta \neq \lambda$. These probability densities will be studied from a different point of view than the one used in~\cite{Zozor17}, precisely because of, as evoked, their relationship to the generalized cosine functions Eq.~\eqref{def:cospq} (see definition~\ref{definition:GTDs} and~Eq.~\eqref{def:rho_pbetalambda} later on).


\subsection{Differential-escort transformations} \label{ssec:definition:diff-esc}

Differential-escort transformations~\cite{Puertas19} make up a family of parametric stretching transformations in such a way that when applied to a probability density function, this one is elevated to a given power while the state axis is stretched to keep invariant the probability in the differential intervals. More precisely:
\begin{definition}[Differential-escort transformation]\label{definition:diff-esc}
Let $\pdf: \support \longrightarrow \Rset_+$ be a probability density and $\alpha$ a real parameter. Let $x_0 \in \support$, and $y_0$ be real numbers. The $\alpha-$order differential-escort transformation of $\pdf$ is a probability density $\descort[\alpha]{\pdf}: \support_\alpha \longrightarrow \Rset_+$ defined by
\begin{equation*}
\descort[\alpha]{\pdf}(y) = \left[ \pdf(x(y)) \right]^\alpha,
\end{equation*}
where $x(y)$ is the inverse function of the one-to-one function $y: \support \longrightarrow \support_\alpha = y(\support)$ defined by
\begin{equation*}
y(x) = y_0 + \int_{x_0}^x \left[ \pdf(t) \right]^{1-\alpha} dt
\end{equation*}
with the assumption that the integral converges in a neighborhood of $x_0$. Since the choice of $x_0$ and $y_0$ only implies a translation, we will consider $y_0 = x_0 = 0$ without loss of generality.
\end{definition}
Particularly, these transformations correspond with the so-named \textit{combinated Sundman transformations}, but choosing the parameters to impose the \textit{probability conservation} in the differential intervals. Sundman transformations appear in the study of three-body problem~\cite{Sundman1913} and found recently applications in bubble dynamics~\cite{Kudryashov14, Gordoa21}, ultrasonic waves~\cite{Gordoa21}, or Rayleigh equation~\cite{Kudryashov15}.

In the following lines we summarize some of the basic properties of the transformation $\descort[\alpha]{}$ proved in~\cite{Puertas19, Zozor17}.

First, it is important to remark that these transformations can be composed in a simple way. Indeed, the composition law is naturally inherited from the power of the real numbers,
\begin{property}[Composition property, neutral element, inversion]\label{prop:composition} 
First, one can easily check that the transformation reduces to the identity for $\alpha = 1$,
\[
\descort[1]{\pdf} = \pdf
\]
Then, given any couple of real numbers $\alpha, \, \alpha',$ 
\[
\descort[\alpha]{} \circ \descort[\alpha']{} = \descort[\alpha']{} \circ \descort[\alpha]{} = \descort[\alpha \alpha']{}.
\]
where $\circ$ is the composition operator. As a consequence, for $\alpha \neq 0$ the transformation is invertible and the inverse transformation is given by
\[
\descort[\alpha]{}^{-1} = \descort[\alpha^{-1}]{}.
\]
\end{property} 
Due to the composition property, and excluding the absorbent element $\alpha = 0$, $\left( \{ \descort[\alpha]{} \}_{\alpha \in \Rset^*} \, , \, \circ \right)$ has an abelian group structure. The composition property is a key property involved in the proof of the generalized informational inequalities we will provide in the sequel.

Another important property is the fact that the length of the support of the differential-escort transform $\descort[\alpha]{}$ of any probability density $\pdf$ is given by the $L_{1-\alpha}$ norm of the original density $\pdf$ to the $1-\alpha$~\cite{Puertas19}, as summarized in the following property. 
\begin{property}[Stretching effect on the length of the support]\label{prop:support_length}
The length of the  support $\support_\alpha$ of the differential-escort density $\descort[\alpha]{\pdf}$ is given by 
\begin{equation*}
L(\support_\alpha)=\int_\support [\pdf(t)]^{1-\alpha}\,dt.
\end{equation*}
\end{property}

The differential-escort transformation also has the remarkable property of conserving scaling effect, the scaling being affected by a power in the transformed domain, as summarized below:
\begin{property}[Scaling Property]\label{prop:scaling}
The scaling transformation and the differential-escort transformation can be commuted thanks to raising the scaling to a power: Given $\kappa \in \Rset_+$, and denoting the scaling operation as $\pdf_{(\kappa)}(x) = \kappa \pdf(\kappa x)$,  its $\alpha-$th order \textit{differential-escort} distribution is the scaling transformed of $\descort[\alpha]{\pdf}$ with  the scaling factor $\kappa^\alpha$, 
\begin{equation*}
\descort[\alpha]{\pdf_{(\kappa)}} = \descort[\alpha]{\pdf}_{(\kappa^\alpha)}
\end{equation*}
\end{property}

In the following properties, a regular behavior with the main informational measures considered in this work is highlighted. Specifically, the Rényi entropy power and the  $(p,\lambda)-$Fisher information of the differential-escort transform of a probability density remarkably are the entropy power and the generalized Fisher information themselves, respectively, up to a power and with modified parameters. In detail:
\begin{property}[Stretching effect on  Rényi entropies]
\label{prop:Renyi-esc}
Given $\lambda \in \Rset$, $\alpha \in \Rset$ and  a probability density $\pdf$, the following identity is satisfied,
\begin{equation}\label{Eq:lambda_alpha}
N_{\lambda}\left[ \descort[\alpha]{\pdf} \right] = \left( N_{\lambda_\alpha}[\pdf] \right)^\alpha, \quad \lambda_\alpha = 1 + \alpha \, (\lambda-1).
\end{equation}
\end{property}
\noindent Let us remark that, up to a power, the Shannon entropy power is invariant by the differential-escort transform ($\lambda = 1 \Rightarrow \lambda_\alpha = 1$).

\begin{property}[Stretching effect on the  $(p,\lambda)-$Fisher information]\label{prop:Fisher-esc}
Given $p > 0$,  $\lambda > \frac1{1+p^*}$, $\alpha > 0$ and  a probability density $\pdf$, the following identity is satisfied
\begin{equation*} 
\phi_{p,\lambda}\left[ \descort[\alpha]{\pdf} \right] = |\alpha|^\frac1\lambda \left( \phi_{p,\alpha\lambda}[\pdf] \right)^\alpha.
\end{equation*}
\end{property}

The bi-parametric Stam inequality given Eq.~\eqref{ineq:bip_Stam} was precisely generalized to the three parametric one given in Eq.~\eqref{ineq:trip_Stam} thanks to properties~\ref{prop:Renyi-esc} and~\ref{prop:Fisher-esc}, which made it possible to decouple the Rényi entropy parameter from the second one of the  $(p,\lambda)-$Fisher information~\cite{Zozor17}.

We evoked in the introduction that the differential-escort transformation changes continuously the weight of the tail of the density, if any, and can even make the support of the transformed probability density to be compact (even for an initial distribution with non-compact support). This is summarized in the following Lemma for power-law-decaying densities~\cite{Puertas19}:
\begin{lemma}[Stretching effect on the tail of power law decaying densities]\label{Lemma:tail}
Let $\pdf > 0$ on $\support = \Rset$ be a probability density such that the tails decrease as a power law, $\O\left( |x|^{-\eta} \right)$ for some $\eta > 1$. Then, for $\alpha > \alpha_c  = \frac1{\eta^*}$, the tails of the transformed distribution $\descort[\alpha]{\pdf}$ decrease as $\O\left(|y|^{- \, \frac{\eta\alpha}{1+\eta(\alpha-1)}}\right)$. On the other hand, for $\alpha<\alpha_c$, the distribution $\descort[\alpha]{\pdf}$ has a compact support. Finally,  $\descort[\alpha_c]{\pdf}$ has a non-compact support and an exponential decay.
\end{lemma}
Note that in the limit $\alpha \to +\infty$ the tails of the distribution $\descort[\alpha]{\pdf}$ tend to decrease as $\O\left(|y|^{-1}\right),$ which is the limit case of a non-integrable tail. Note also that, clearly, the supremum 
\[
\sup \left\{ p > 0, \: \left\langle \left| x \right|^p \right \rangle_{\descort[\alpha]{\pdf}} < +\infty \right\},
\]
for $\alpha\in[\alpha_c,\infty)$, under the assumptions of Lemma~\ref{Lemma:tail}, is a decreasing function of  $\alpha$ since the weight of the tails is increasing with $\alpha$. Now, if we come back to the composition property~\ref{prop:composition}, it follows that the operator $\descort[\alpha]{}$ has the effect of increasing (resp. decreasing) the weight of the tails of a power-law density when $\alpha>1$ (resp. $\alpha<1$). It is worth mentioning that this behavior is  the opposite of that  of the \textit{usual} escort transformation,
\begin{equation}\label{eq:usual_escort}
\escort[\alpha]{\pdf}(x) = \frac{\left[ \pdf(x) \right]^\alpha}{\displaystyle \int_\Rset \left[ \pdf(t) \right]^\alpha dt},
\end{equation}
for  which the weight of the tails is reduced when $\alpha>1$, and increased when $\alpha\in(1/\beta,1)$, because the tail is crushed\footnote{Obviously, $\pdf(x) < 1$ for large values of $|x|$, so that  for large $|x|$, $\left[ \pdf(x) \right]^\alpha < \pdf(x)$ for any $\alpha>1.$} when $\alpha > 1$. The difference between the escort $\escort[\alpha]{}$ and the differential-escort transformation $\descort[\alpha]{}$ lies in the fact that the support of  $\escort[\alpha]{\pdf}$ remains always invariant, and the normalization is imposed by dividing the powered density by its $L^1-$norm. In contrast, in the differential-escort transformation, the normalization is carried out by imposing the conservation of the probability in any differential interval through the change of variables given in definition~\ref{definition:diff-esc}. Then, although the tails seem crushed when one applies the transformation $\descort[\alpha]{}, \: \alpha >1$, the length of each differential interval of the tails is increased to keep invariant the probability. This fact provokes the previously discussed increase of the tail weight when applying $\descort[\alpha]{}, \: \alpha>1$.

Another significant difference between usual and differential-escort transformations is the fact that, while the differential-escort transformations are well-defined for any $\alpha\in\Rset$, the escort distribution $\escort[\alpha]{\pdf}$ is only well-defined when $\displaystyle \int_\Rset \left[ \pdf(x) \right]^\alpha dx < \infty$, which is not necessarily satisfied for any $\alpha$. Just think about a power-law; indeed, the escort density $\escort[\alpha]{\pdf}$ of a power-law is a power-law whatever $\alpha$, provided that the power is integrable. Moreover, in such a case, an exponential behavior is never obtained (and \textit{all} power moments are finite only when $\alpha \to +\infty$). On the opposite, when $\alpha<1$, the weight of the tails of the differential-escort density $\descort[\alpha]{\pdf}$ is reduced up to achieve an exponentially decaying density for the critical parameter $\alpha_c$. Below, i.e., for $\alpha < \alpha_c$, the support of the differential-escort density becomes compact, the differential-escort density becoming uniform when $\alpha = 0$ on the unit length support\footnote{The support $(0,1)$ in the case $\alpha=0$ lies with the choice  $x_0=y_0=0$ in the definition~\ref{definition:diff-esc};  remember that another choice of $x_0$ and $y_0$ only implies a translation.} $(0,1)$. For completeness, let us mention that when $\alpha < 0,$ the support of $\descort[\alpha]{\pdf}$  is compact but the density diverges in its borders/adherence (again, under the assumptions of  Lemma~\ref{Lemma:tail} for $\pdf$).


\section{Cumulative moments}\label{Sec_CumMom}


\subsection{Definition and general properties}

The informational quantifiers mentioned in section~\ref{Sec_Preliminary} are functionals of the probability density and its derivative, but there is no explicit dependence in the cumulative density function $\cdf$ or its complementary~$\ccdf$,
\begin{equation}\label{def:cum_func}
\cdf(x)=\int_{-\infty}^x \pdf(t)\,dt, \qquad  \ccdf(x) = 1 - \cdf(x) =\int_x^{\infty} \pdf(t)\,dt
\end{equation}
However, an entropy-like measure, called cumulative residual entropies defined as a functional solely of, $\cdf(x)$, was introduced for reliability engineering and image processing purposes in~\cite{Rao04} as
\begin{equation*}
\E[\pdf] = -\int_{\Rset} \ccdf(x)\log\ccdf (x)\,dx.
\end{equation*}
Nevertheless, as far as we know, definitions of functionals involving both the probability density function and the corresponding cumulative density function are not usual in the literature.

Following the line of~\cite{Rao04}, we introduce in this section a family of generalized moments $\mu_{p,\lambda}$, defined as the expected value of a deformed cumulative function involving a power of the probability density. As we will see, these moments are remarkably connected to usual moments thanks to a differential-escort transformation of the density. Let us first introduce the weighted cumulative function, defined through the usual escort density.
\begin{definition}[Weighted cumulative function]\label{def:cum_func}
Let $\gamma \in \Rset$, and a probability density $\pdf$. We define the  weighted cumulative function $\wcdf{\gamma}$ as
\begin{equation}\label{def:weighted_cum_func}
\wcdf{\gamma}(x)=\int_{-\infty}^x \left[ \pdf(t) \right]^{1-\gamma} \, dt.
\end{equation}
\end{definition}
\noindent The usual cumulative function is recovered when $\gamma = 0$, while
in the case $\gamma=1,$ $\wcdf{\gamma}(x)$ must be interpreted as the length of the support intersected with the interval  $(-\infty,x)$, possibly infinite.

In the sequel, only the quantity $\displaystyle \int_0^x[\pdf(t)]^{1-\gamma} dt \equiv \wcdf{\gamma}(x) -\wcdf{\gamma}(0)$, well-defined for any $\gamma \in \Rset$ and any bounded density $\pdf$, will be relevant.

Then, we define the $(p,\gamma)-$th generalized cumulative moment $\mu_{p,\gamma}[\pdf]$ of a density $\pdf$  as follows:
\begin{definition}[Generalized and generalized absolute cumulative moments and cumulative deviations]\label{definition:mu_plambda}
Given two real numbers $p > 0$, and $\gamma \in \Rset$, we define the generalized absolute cumulative moment of $(p,\gamma)-$th order as\footnote{Defining central moments mimicking usual ones is not a trivial task. The discussion of such an extension is postponed to the Appendix~\ref{app:central}.}
\begin{equation*}
\mu_{p,\gamma}[\pdf] = \int_\Rset \left| \int_0^x \left[ \pdf(t) \right]^{1-\gamma} dt \right|^p \pdf(x) \, dx = \left\langle \Big| \wcdf{\gamma}(x) - \wcdf{\gamma}(0) \Big|^p \right\rangle_{\! \pdf}.
\end{equation*}
Accordingly, the associated \textit{cumulative deviation} is defined as
\[
\sigma_{p,\gamma}[\pdf] = \mu_{p,\gamma}[\pdf]^{\frac1{p\gamma}}
\]
In the same line, when $p$ is a natural number, we define the generalized cumulative moments as
\begin{equation*}
\widetilde{\mu}_{p,\gamma}[\pdf] = \int_\Rset \left( \int_0^x \left[ \pdf(t) \right]^{1-\gamma} dt \right)^p \pdf(x) \, dx = \left\langle \Big( \wcdf{\gamma}(x) - \wcdf{\gamma}(0) \Big)^p \right\rangle_{\! \pdf}.
\end{equation*}
\end{definition}
\noindent Trivially the usual moments and absolute moments are recovered when $\gamma \to 1$, $\mu_{p,1}[\pdf]  = \mu_p[\pdf]$ \ $\widetilde{\mu}_{p,1}[\pdf] = \widetilde{\mu}_p[\pdf]$.

The absolute cumulative moments and deviations follow some basic properties with respect to their parameters, differential-escort transformation, and scaling, summarized below.
\begin{property}[Order relation]\label{prop:order_rel}
Given a probability density $\pdf$ and a real $\gamma$, function $p \mapsto \sigma_{p,\gamma}[\pdf]$ is increasing, that is, for any real numbers $0 < a < b$, 
\begin{equation*}
\sigma_{a,\gamma}[\pdf] \leqslant \sigma_{b,\gamma}[\pdf],
\end{equation*} 
\end{property}
\begin{proof}
Using Hölder inequality one easily obtains $\mu_{a,\gamma}[\pdf] \leqslant \left[ \mu_{ap,\gamma}[\pdf] \right]^{\frac1p}$ for any $p > 1$. We conclude by taking $b = p \, a$.
\end{proof}

As evoked, the absolute cumulative and usual moments are intimately linked through the differential-escort transformation:
\begin{property}[Cumulative moments and differential-escort transformation]\label{prop:mu_p-esc}
Given a probability density $\pdf$, and  two real numbers $\alpha \in \Rset$ and $p >0$,  the following identities are fulfilled
\begin{equation*}
\widetilde{\mu}_p[\descort[\alpha]{\pdf}] = \widetilde{\mu}_{p,\alpha}[\pdf]
\end{equation*}
and
\begin{equation*}
\mu_p\left[ \descort[\alpha]{\pdf} \right] = \mu_{p,\alpha}[\pdf], \qquad \sigma_p\left[ \descort[\alpha]{\pdf} \right] = \left[ \sigma_{p,\alpha}[\pdf] \right]^\alpha
\end{equation*} 
\end{property}
\begin{proof}
The proof trivially follows from definitions~\ref{definition:diff-esc} and~\ref{definition:mu_plambda}.
\end{proof}

On the other hand, as a consequence of the above property~\ref{prop:mu_p-esc} and the composition property~\ref{prop:composition}, one obtains a remarkable relation between cumulative moments and the $\alpha-$order differential-escort transform, in concordance with the Rényi and Fisher properties~\ref{prop:Renyi-esc} and~\ref{prop:Fisher-esc}:
\begin{property}[Stretching and cumulative moments]\label{prop:mu_plambda-esc}
Given three real numbers $p > 0$, $\alpha, \gamma \in \Rset$, the absolute cumulative moments and deviations of a density $\pdf$ satisfy
\begin{equation*}
\widetilde\mu_{p,\gamma}\left[ \descort[\alpha]{\pdf} \right] = \widetilde\mu_{p,\gamma\alpha}[\pdf], \qquad\mu_{p,\gamma}\left[ \descort[\alpha]{\pdf} \right] = \mu_{p,\gamma\alpha}[\pdf], \qquad \sigma_{p,\gamma}\left[ \descort[\alpha]{\pdf} \right] = \left[ \sigma_{p,\gamma\alpha}[\pdf] \right]^\alpha.
\end{equation*}
\end{property}
\begin{proof}
Given a couple of real numbers $\alpha, \gamma$, from the composition property~\ref{prop:composition} we have $\widetilde{\mu}_p\left[ \descort[\gamma \alpha]{\pdf} \right] = \widetilde{\mu}_p\left[ \descort[\gamma]{} \circ \descort[\alpha]{\pdf} \right]$. Now, using  property~\ref{prop:mu_p-esc}, we obtain $\widetilde{\mu}_p\left[ \descort[\gamma]{} \circ \descort[\alpha]{\pdf} \right] = \widetilde{\mu}_{p,\gamma}\left[ \descort[\alpha]{\pdf} \right]$, as well as, $\widetilde{\mu}_p\left[ \descort[\gamma \alpha]{\pdf} \right] = \widetilde{\mu}_{p,\gamma \alpha}[\pdf]$. So, necessarily, $\widetilde{\mu}_{p,\gamma}\left[ \descort[\alpha]{\pdf} \right] = \widetilde{\mu}_{p,\gamma\alpha}[\pdf]$, which closes the proof.  The proof for $\mu_{p,\lambda}$ follows step by step the same line.
\end{proof}

As for usual absolute moments, the absolute cumulative moments satisfy extended properties with respect to scaling changes. These are the consequences of properties~\ref{prop:scaling} and~\ref{prop:mu_p-esc}, summarized as follows:
\begin{property}[Scaling and absolute cumulative moments]\label{prop:scaling-mu_plambda}
Given two real number $p > 0$, and $\gamma \in \Rset^*$, and the scaling transformed density $\pdf_{(\kappa)}(x) = \kappa \, \pdf(\kappa x)$, $\kappa>0$, of any probability density $\pdf$, their absolute cumulative  deviations are related by
\begin{equation*}
\widetilde\mu_{p,\gamma}\left[ \pdf_{(\kappa)} \right] = \kappa^{-p\gamma} \, \widetilde\mu_{p,\gamma}[\pdf],\qquad\sigma_{p,\gamma}\left[ \pdf_{(\kappa)} \right] = \kappa^{-1} \sigma_{p,\gamma}[\pdf].
\end{equation*}
\end{property}
\begin{proof}
Given the real numbers $\alpha \in \Rset^*,\,p>0$ and $ \kappa > 0$, from property~\ref{prop:scaling}, and  using the relation $\sigma_p\left[ \pdf_{(\kappa)} \right] = \kappa^{-1}\sigma_p[\pdf]$, one has $\sigma_p\left[ \descort[\alpha]{\pdf_{(\kappa)}} \right] = \sigma_p\left[ \descort[\alpha]{\pdf}_{(\kappa^\alpha)} \right] = \kappa^{-\alpha} \sigma_p\left[ \descort[\alpha]{\pdf} \right].$ Now, using property~\ref{prop:mu_p-esc} in both sides of the  last equality one obtains $\left[ \sigma_{p,\alpha}[\pdf_{(\kappa)}] \right]^\alpha = \left[ \kappa^{-1} \sigma_{p,\alpha}[\pdf] \right]^\alpha,$ which ends the proof. The proof for $\widetilde \mu_{p,\lambda}$ follows the same line. 
\end{proof}

To conclude, let us highlight that Lemma~\ref{Lemma:tail} and property~\ref{prop:mu_p-esc} allows to prove the existence of a critical value $\gamma_c$ such that $\mu_{p,\gamma} < \infty$ for any $p \geqslant 0, \, \gamma \in [0 , \gamma_c]$:
\begin{lemma}[Finite values of the (absolute) cumulative moments for heavy-tailed distributions]\label{Lemma:mu_plambda-finite}
Given $\pdf(x)$ a bounded probability density which decays as a power-law, $\pdf(x) \sim |x|^{-\eta}$ when $|x| \to +\infty$, then for a fixed $\gamma \leqslant \gamma_c = \frac1{\eta^*},$ $$\mu_{p,\gamma}[\pdf] < \infty, \quad\forall p \geqslant 0.$$ As a consequence, also $$\widetilde{\mu}_{p,\gamma}[\pdf] < \infty, \quad\forall p \in \Nset.$$
\end{lemma}
\begin{proof}
From Lemma~\ref{Lemma:tail} we have that $\descort[\gamma]{\pdf}$ has a compact support for all $\gamma < \gamma_c,$ so that $\mu_p\left[\descort[\gamma]{\pdf} \right] < \infty$ for all $p \geqslant0,$ and thus $\mu_{p,\gamma}[\pdf] < \infty$ for all $p \geqslant 0$. In the critical case, $\descort[\gamma_c]{\pdf}$ has an exponential decay as stated in Lemma~\ref{Lemma:tail}, leading to $\mu_p\left[ \descort[\gamma_c]{\pdf} \right] = \mu_{p,\gamma_c}[\pdf] < \infty,$ for all $p\geqslant0.$
\end{proof}


\subsection{The case $\gamma = 0$}\label{ssec:lambda=0}
Let us here specifically analyze the case $\gamma = 0$. The first relevant observation is that  $\mu_{q,0}[\pdf]$ is completely determined by the \textit{usual} cumulative function evaluated in zero\footnote{The point where the cumulative function is evaluated depends on the definition of cumulative moments, which is linked with the choice of the parameters $x_0 = y_0 = 0$.}.  By definition
\begin{equation*}
\mu_{p,0}[\pdf] = \int_\Rset \left| \int_0^x \pdf(t) \, dt \right|^p \pdf(x) \, dx = \int_\Rset \left| \cdf(x) - \cdf(0) \right|^p \cdf'(x) \, dx.
\end{equation*}
After simple algebraic manipulations, one obtains:
\begin{equation*}
\mu_{p,0}[\pdf] = \frac{\left[ \ccdf(0) \right]^{p+1} + \left[ \cdf(0) \right]^{p+1}}{p+1}.
\end{equation*}

Let us here compute the cumulative moments for the exponential $\pdf_1(x) = \minpdf{1}{+\infty}{}(x) \propto e^{-x} \one_{\Rset_+}(x)$ and $q-$exponential densities $\pdf_q(x) = \minpdf{2-q}{+\infty}{}(x) \propto \exp_q(-x) \one_{\Rset_+}$ where $\exp_q$ is defined Eq.~\eqref{def:q-exp}.

First, let us mention the apparently poorly known fact that $q-$exponential densities (with a scaling operation) can be characterized as the differential-escort of the exponential one~\cite{Puertas19},
\begin{equation}\label{eq:e-e_q}
\descort[\alpha]{\pdf_1} = \left( \pdf_q \right)_{(\alpha)} \qquad \mbox{with} \quad q = 2 - \frac1\alpha.
\end{equation}
Insofar, it is important to remark that the differential-escort transformation of a Gaussian density is not a $q-$Gaussian one, but involves some special functions as discussed in reference~\cite{Zozor17} and as will be recalled later in the paper. In fact, as advanced in the introduction, the transformation $\descort[\alpha]{}$ of the generalized Gaussian distribution $\gauss_{p,\lambda}$ is connected to generalized trigonometric functions. This will be discussed in more detail in the next Section~\ref{Sec_densities}.


\subsection{Reconstruction of the probability density via cumulative moments}\label{ssec:cauchy}

Beyond the study of the cumulative moments in the informational-theoretic framework, these generalized moments can find application in terms of probability density reconstruction. Under certain conditions, a series of standard moments allows to formally reconstruct a probability density~\cite{Shohat1943}. More precisely, a numerable sequence $\left\{ m_i \right\}_{i=1}^\infty$ is said to be a moment-sequence when there exists a probability density $\pdf$ with $\widetilde{\mu}_i[\pdf] = m_i$. In addition, it is known that when a moment-sequence satisfies the so-called Carleman's condition~\footnote{Carleman's condition is satisfied when $\displaystyle \sum_{i=1}^\infty \frac1{\sqrt[2i]{m_{2i}}} = \infty.$} then the sequence $\left\{ m_i \right\}_{i=1}^\infty$ uniquely characterizes the probability density $\pdf,$ or equivalently, the associated Hamburger problem is determinated~\cite{Akhiezer_Moment_Problem}. Let us point out that Carleman's condition is a sufficient condition. As far as we know, unicity sufficient and necessary conditions  it still an open problem.

When a probability density is heavy-tailed, it does not admit a numerable set of well-defined moments. An approach using the moments of the usual escort density Eq.~\eqref{eq:usual_escort} has been studied by Tsallis et al.~\cite{Tsallis09} to generalized the density reconstruction from the ``moments''. However, they conclude that the moments of the usual escort transformation do not allow the recovery of a heavy-tailed density satisfactorily. More precisely, they argue that as far as the so-called $q-$Fourier transform has not a unique inverse function, then a probability density $\pdf$ cannot be fully characterized by the set of the moments of its escort transform (given in Eq.~\eqref{eq:usual_escort}). As an alternative, we propose to have recourse to cumulative moments to characterize such heavy-tailed densities. More precisely, the set $\left\{ \{\gamma\} \, , \, \{m_i\}_{i=1}^\infty \right\}$  with $\gamma \neq 0$  characterizes a probability density $\pdfcm$, where the role of the parameter $\gamma$ is to ``regulate'' weight of the tail. In particular, the cumulative moments of $\pdfcm$ are given by $\widetilde\mu_{\gamma,i}[\pdfcm] = m_i$,  while not necessarily all its standard moments are finite.
\begin{theorem}\label{th:heavy-Hauss}
Let $\{m_i\}_{i=1}^\infty$ be a moment-sequence which uniquely characterizes a probability density $\pdf$. Then, for any $\gamma \neq 0$, there exists a probability density $\pdfcm$ (depending on $\gamma$) such that $\mu_{\gamma , i}[\pdfcm] = m_i.$ Moreover, $\left\{ \{\gamma\} \, , \, \{m_i\}_{i=1}^\infty \right\}$ uniquely characterizes $\pdfcm$, which is given by $\pdfcm = \descort[\gamma^{-1}]{\pdf}$.
\end{theorem}
\begin{proof}
The result is an immediate consequence of the group structure of differential-escort transformations given Property~\ref{prop:composition} and Property~\ref{prop:mu_p-esc} that give $
m_i = \widetilde\mu_i\left[ \descort[\gamma]{} \circ \descort[\gamma^{-1}]{\pdf} \right] = \widetilde\mu_{\gamma,i} \left[ \descort[\gamma^{-1}]{\pdf} \right]$
\end{proof}
This result not only provides a richer characterization of a probability density but also allows to characterize heavy-tailed probability density functions. Indeed,  by Lemma~\ref{Lemma:mu_plambda-finite}, for any bounded probability density with an algebraic decay 
there exists a critical value $\gamma_c$  such that the set $\{\widetilde\mu_{\gamma,i}[\pdf]\}_{i=1}^\infty$ is well defined for any $\gamma \leqslant \gamma_c$.

Let us apply this result to the Cauchy probability density  for which the critical parameter in Lemma~\ref{Lemma:mu_plambda-finite} is given by $\gamma_c = \frac12.$ By definition,
\begin{equation}
\widetilde\mu_{\frac12 , 2i}[\gauss_{2,0}] = \frac1{\pi^{1+i}} \int_\Rset \left( \int_0^x \frac1{\sqrt{1+x^2}} \right)^{2i} \frac{dx}{1+x^2} = \frac2{\pi^{1+i}} \int_0^{+\infty} \frac{\arsinh^{2i}(x)}{1+x^2}\,dx,
\end{equation}
which correspond to arcsinh-like moments. After the successive changes of variable $x = \sinh s$ and $z = e^s$, and from~\cite[Eq.~4.271-6]{Gradshteyn15}, we obtain for any natural values of $i$
\begin{equation}\label{eq:Euler_numbers}
\widetilde\mu_{\frac12 , 2i}[\gauss_{2,0}] = \frac4{\pi^{1 + i}} \int_{1}^{+\infty} \, \frac{\log^{2i}(z)}{1+z^2}\,dz = \left(\frac{\pi}{4}\right)^i \left| E_{2i} \right|, 
\end{equation}
where $E_{2i}$ denotes with the Euler numbers~\cite{Gradshteyn15}. The odd order cumulative moments are trivially zero by parity, i.e., $\widetilde\mu_{\frac12 , 2i+1}[\gauss_{2,0}]=0,$ from where the following expression is fulfilled $$\widetilde\mu_{\frac12,i}[\gauss_{2,0}]=\left(\frac{\sqrt{\pi}}{2}\right)^i|E_i|$$ 
A simple study allows to verify that the sequence $\left\{ \left(\frac{\sqrt{\pi}}{2}\right)^i|E_i|\right\}_{i=1}^{+\infty}$ satisfies the Carleman's condition. Hence, Theorem~\ref{th:heavy-Hauss} allows to insure that the Cauchy probability density is uniquely characterized by the sequence $\{\widetilde\mu_{\frac12 , i}\}_{i = 1}^\infty.$ More precisely, $\widetilde\mu_{\frac12 , i}[\gauss_{2,0}] = \widetilde\mu_{i}[\descort[\frac12]{\gauss_{2,0}}] = \widetilde\mu_{i}[\gauss_{2,0,-1}],$ where $\gauss_{2,0,-1}$ corresponds to the hyperbolic secant probability density, which has an exponential decay.
In summary, the series $\left\{ \{1\} \, , \,\left\{ \left(\frac{\sqrt{\pi}}{2}\right)^i|E_i| \right\}_{i = 1}^\infty \right\}$ uniquely characterizes the hyperbolic secant density, and the series $\left\{ \left\{ \frac12 \right\} \, , \, \left\{ \left(\frac{\sqrt{\pi}}{2}\right)^i|E_i| \right\}_{i = 1}^\infty \right\}$ uniquely characterizes the Cauchy probability density. 

As mentioned in the introduction of the section, the reconstruction problem of a probability density from the sequence of moments is ill-posed, while the MaxEnt principle is a good alternative to approximate the density with a fixed finite number of moments. Let us also recall the inadequacy of this approach dealing with heavy-tailed distributions. As claimed, theorem~\ref{th:heavy-Hauss} precisely provides the framework to extend the MaxEnt approach, dealing with cumulative moments.
In the following theorem we enlight the equivalence between maximizing Shannon or Rényi entropies and fixing both standard or cumulative moments.
\begin{theorem}\label{th:Cumlative-MaxEnt}
Let $f$ be the probability density which maximizes the Shannon entropy with the first $N$ moments $\widetilde\mu_i[f] = m_i, \; i = 1 , 2 , \ldots , N$ fixed, and let $\gamma>0$ be a real number. It turns that $\pdfcm = \descort[\gamma^{-1}]{f}$ maximizes the Shannon entropy with the first $N$ cumulative moments $\widetilde\mu_{i,\gamma}[\pdfcm] = m_i , \; i = 1 , 2 , \ldots , N$ fixed. Moreover, if $\pdf_\lambda$ is the density maximizing the $\lambda-$Rényi entropy with the first $N$ moments $\widetilde\mu_i[\pdf_\lambda] = m_i, \; i = 1 , 2 , \ldots , N$ fixed,  then $\pdfcm_\lambda = \descort[\gamma^{-1}]{\pdf_\lambda}$ 
is the probability density maximizing the $\lambda_\gamma-$Rényi entropy with the first $N$ cumulative moments $\widetilde\mu_{i,\gamma}[\pdfcm_\lambda] = m_i, \; i = 1 , 2 , \ldots , N$  fixed, and with $\lambda_\gamma = 1 + \gamma \, (\lambda-1)$.
\end{theorem}
\begin{proof}
The proof follows from Properties~\ref{prop:composition}, \ref{prop:Renyi-esc} and \ref{prop:mu_p-esc}.
\end{proof}


\section{Generalized trigonometric densities}\label{Sec_densities}

In this section, we will study the minimizing densities $\minpdf{p}{\beta}{\lambda}$ of the generalized Stam inequality~\eqref{ineq:trip_Stam} and will relate them with the generalized trigonometric functions. In addition, we will compute some of their cumulative moments.


\subsection{The stretched-Gaussian and the generalized cosine} \label{sec_trig}

Let us first highlight the relation between stretched Gaussian densities and generalized trigonometric functions through the $\alpha-$order differential-escort transformation. This is summarized in the following lemma:
\begin{lemma}[Generalized cosine and stretched Gaussian densities]\label{Lemma:GTDs&gplambda}
Consider the stretched Gaussian density $\gauss_{p,\lambda}: \support_{p,\lambda} \longrightarrow \Rset_+, p^*>0, \; \lambda > 1 -p^*, \; \lambda \neq 1,$ and a real number $\alpha \neq 1$. Then, the differential-escort density of the stretched Gaussian, $\descort[\alpha]{\gauss_{p,\lambda}}: (\support_{p,\lambda})_\alpha \longrightarrow \Rset_+$ is given by
\begin{equation}\label{eq:gplalpha}
\descort[\alpha]{\gauss_{p,\lambda}}(y) = \left\{\begin{array}{lll}
a_{p,\lambda}^\alpha \, \left[ \cosh_{\frac{1-\lambda}{1-\alpha},p^*}\left(  k_{p,\lambda,\alpha} y \right) \right]_+^{\frac{\alpha}{\alpha-1}} & \mbox{if} &  \lambda < 1,
\\[4mm]
a_{p,\lambda}^\alpha \, \left[ \cos_{\frac{\lambda-1}{\alpha-1},p^*}\left(  k_{p,\lambda,\alpha} y \right) \right]_+^{\frac{\alpha}{\alpha-1}} & \mbox{if} & \lambda > 1.
\end{array}\right.
\end{equation}
where $k_{p,\lambda,\alpha} = \frac{|1-\lambda|^\frac1{p^*}}{a_{p,\lambda}^{1-\alpha}},$ and where $a_{p,\lambda}$ is given in Eq.~\eqref{def_aplambda}. The support of this density is given by
\begin{equation*}
(\support_{p,\lambda})_\alpha = \left( - \, \frac{[a_{p,\lambda}]^{1-\alpha}\,\pi_{\xi,p^*}}{2\,|1-\lambda|^\frac{1}{p^*}} \: , \: \frac{[a_{p,\lambda}]^{1-\alpha}\,\pi_{\xi,p^*}}{2\,|1-\lambda|^\frac{1}{p^*}} \right)
\end{equation*}
with
\begin{equation*}
\xi = \left(\frac{\alpha-1}{|1-\lambda|}+\frac{p^*+1}{p^*}\cdot\one_{\Rset_+}(1-\lambda)\right)^{-1}.
\end{equation*}
\end{lemma}
\begin{proof}
The proof essentially follows from definition~\ref{definition:diff-esc}, Eq.~\eqref{def:g_plambda}, and the pythagorean-like relations given by Eqs~\eqref{eq:pyth_rel}-\eqref{eq:pyth_rel_hyperbolic}. For more details see Appendix~\ref{App1}.
\end{proof}
\noindent Note that in the case $\alpha = 1$ the generalized Gaussian densities are trivially recovered from property~\ref{prop:composition}.


Let us mention some remarks concerning Lemma~\ref{Lemma:GTDs&gplambda}.
\begin{remark}\label{remark_integrability_1}
The integrability conditions of the deformed probability density $\descort[\alpha]{\gauss_{p,\lambda}}$ are inherited from those of the stretched Gaussian $\gauss_{p,\lambda},$ i.e. $p^* > 0$ and $\lambda > 1 - p^*$. This is because, as discussed in section~\ref{ssec:definition:diff-esc}, the differential-escort transform $\descort[\alpha]{\pdf}$ of any  probability density $\pdf$ is also a probability density (with $L_1-$norm unity) whatever $\alpha \in \Rset.$   
\end{remark}
\begin{remark}\label{remark_compact_support_1}
By virtue of Lemma~\ref{Lemma:tail} and the definition of stretched Gaussian $\gauss_{p,\lambda}$ Eq.~\eqref{def:g_plambda}, one finds, for $\lambda < 1$, that the deformed probability density $\descort[\alpha]{\gauss_{p,\lambda}}$ has a compact support for any $\alpha < 1 + \frac{ \lambda - 1 }{ p^* }.$ At the opposite, i.e., for $\alpha \geqslant 1 + \frac{ \lambda - 1 }{ p^* }$, its support is the whole real set and $\descort[\alpha]{\gauss_{p,\lambda}}$ has an exponential decay for $\alpha = 1 + \frac{ \lambda - 1 }{ p^* },$ and power-law decay for $\alpha > 1 + \frac{ \lambda - 1 }{ p^* }.$ In the case $\lambda > 1,$ using property~\ref{prop:support_length}, one finds that the support of $\descort[\alpha]{\gauss_{p,\lambda}}$ is compact for any $\alpha < \lambda$ and the whole real set for $\alpha \geqslant \lambda$ where, for $\alpha = \lambda$ the deformed density $\descort[\alpha]{\gauss_{p,\lambda}}$ has an exponential decay,  and a power-law decay for $\alpha > \lambda.$
\end{remark}
\begin{remark}\label{remark_bounded_1}
Note that, since any stretched Gaussian density is bounded, any of its differential-escort density $\descort[\alpha]{\gauss_{p,\lambda}}$ for $\alpha \geqslant 0$ remains bounded. In the opposite, when $\alpha <0$,  $\descort[\alpha]{\gauss_{p,\lambda}}$ is no longer bounded, but divergent in the borders of its compact support. This fact is easy to check from definition~\ref{definition:diff-esc}, and the definition of the stretched Gaussian densities Eq.~\eqref{def:g_plambda}.
\end{remark}

The above lemma motivates the definition of what we name \textit{generalized trigonometric densities} (GTDs).
\begin{definition}[Generalized trigonometric densities (GTDs)]\label{definition:GTDs}
Given real numbers $p^* > 0$, and $\lambda, \beta$ such that $1-\lambda+\beta \neq 0$ and $\sign(1-\lambda+\beta) = \sign\left(\frac{1-\lambda}p+\beta\right)$ with $\sign$ the sign function,  $\lambda \neq 1, \lambda \neq \beta$, we define the \textit{ generalized trigonometric density} $\gtd{p}{\beta}{\lambda}: \support_{p,\beta,\lambda} \longrightarrow \Rset_+$ as follows:
\begin{equation}\label{def:GTDs}
\gtd{p}{\beta}{\lambda}(y) = \left\{\begin{array}{lll}
a_{p,\beta,\lambda} \, \left[ \cosh_{\frac{1-\lambda}{\beta - \lambda},p^*}\left( \kappa_{p,\beta,\lambda} \,y \right) \right]^\frac1{\lambda-\beta} & \mbox{if} & \lambda<1,\\[4mm]
a_{p,\beta,\lambda}\, \left[ \cos_{\frac{\lambda-1}{\lambda-\beta},p^*} \left( \kappa_{p,\beta,\lambda} \, y \right) \right]^\frac1{\lambda-\beta} & \mbox{if} & \lambda>1,
\end{array}\right.
\end{equation}
where
\begin{equation}\label{def:TDs_K}
\kappa_{p, \beta,\lambda} = \left| \frac{\lambda-1}{1+\beta-\lambda} \right|^\frac1{p^*},
\end{equation}
and the normalization constant is given by
\begin{equation}\label{def:a_pbetalambda}
a_{p,\beta, \lambda} = a_{p,\frac{\beta}{1+\beta-\lambda}} = \left\{ \begin{array}{lll}
\displaystyle  \frac{p^* \; \kappa_{p,\beta,\lambda}}{2 \, B\left( \frac1{p^*} , \frac{\beta \, \sign(1 + \beta - \lambda)}{|1 - \lambda|} + \frac1p \, \one_{\Rset_+} \left( \frac{1 - \lambda}{1 + \beta - \lambda} \right) \right)} & \mbox{if}  & \lambda \neq 1,\\[2mm]
\frac{p^*}{2\Gamma(\frac1{p^*})}, & \text{if} &  \lambda = 1,
\end{array}\right.
\end{equation}
The support $\support_{p,\beta,\lambda}$ of this density is  expressed as
\begin{equation*}
\support_{p,\beta,\lambda} = \left( - \frac{\pi_{\nu,p^*}}{2 \kappa_{p,\beta,\lambda}} \: , \: \frac{\pi_{\nu,p^*}}{2 \kappa_{p,\beta,\lambda}} \right)
\end{equation*}
with
\begin{equation*}
\nu = \left(\frac{p^* + 1}{p^*} \: \one_{\Rset_+} \left( \frac{1-\lambda}{1 + \beta - \lambda} \right) + \frac{\lambda - \beta}{|1 - \lambda|} \, \sign(1 + \beta - \lambda) \right)^{-1}.
\end{equation*}
\end{definition}

Definition~\ref{def:GTDs} allows to rewrite the Lemma~\ref{Lemma:GTDs&gplambda} as follows
\begin{corollary}[Trigonometric and stretched Gaussian densities]\label{coroll}
Given real numbers $p, \lambda, \beta$ such that $p^*>0$, $1+\beta-\lambda\neq0$ and $\sign(1-\lambda+\beta) = \sign\left(\frac{1-\lambda}p+\beta\right)$,  $\lambda \neq 1, \lambda \neq \beta$ is fulfilled,
\begin{equation*}
\descort[\overline\alpha^{\, -1}]{\gauss_{p,\lambda_0}} = \kappa \, \gtd{p}{\beta}{\lambda}\left(\kappa \;y\right) = \left( \gtd{p}{\beta}{\lambda} \right)_{(\kappa)}(y),
\end{equation*}
where $\overline\alpha = 1 + \beta - \lambda, \lambda_0 = \frac{\beta}{1+\beta-\lambda} \neq 1$ and $\kappa = a_{p,\beta,\lambda}^{\frac{\lambda-\beta}{1+\beta-\lambda}}$.
\end{corollary}
\begin{proof}
The proof trivially follows from Lemma~\ref{Lemma:GTDs&gplambda} and definition~\ref{definition:GTDs}. 
\end{proof}
\noindent The generalized Gaussian densities $\gauss_{p,\lambda}$ are recovered for $\lambda=\beta,$ which is obvious because in this case $\overline \alpha=1,$ and  from property~\eqref{prop:composition}.

As previously discussed in Section~\ref{ssec:Trig},  let us emphasize that the support of the \textit{cosine} type densities in Eqs.~\eqref{eq:gplalpha} and~\eqref{def:GTDs} is not finite for all values of the parameters, nor in the hyperbolic case the support is always the entire real set $\Rset$. Note that making the distinction between the cosine and the hyperbolic cosine functions is superfluous, due to the duality relations given in Eqs.~\eqref{eq:duality} and~\eqref{eq:duality2}.


Let us come back to the previous comments, in light of the corollary.
\begin{remark}\label{remark_integrability_2}
Simple algebraical manipulations allows to show that the integrability conditions $p^* > 0$ and $\lambda > 1 - p^*$ in Lemma~\ref{Lemma:GTDs&gplambda} turn to be $p^* > 0$, $\sign(1 - \lambda + \beta) = \sign\left( \frac{ 1 - \lambda }{ p } + \beta \right)$ when $\lambda$ is substituted by $\lambda_0$ defined in Corollary~\ref{coroll}.Specifically,
\[\lambda_0 - 1 + p^* = \frac{ \lambda - 1 }{ 1 + \beta - \lambda } + p* = \frac{ (1 - \lambda) (p^* - 1) +p^* \beta }{ (1 + \beta - \lambda )p^*} = \frac{ \frac{ 1 - \lambda}{ p } + \beta }{ 1 + \beta - \lambda } \]
from where 
\[
\lambda_0 - 1 + p^*  > 0 \Leftrightarrow  \sign(1+\beta-\lambda)=\sign \left( \frac{ 1 - \lambda }{ p } + \beta \right) 
\]
\end{remark}
\begin{remark}\label{remark_support_length_2}
From the facts pointed out in remark~\ref{remark_compact_support_1}, definition~\ref{definition:GTDs} and corollary~\ref{coroll}, it follows that for $1 + \beta - \lambda > 0$, and $\lambda < 1$ the density $\gtd{p}{\beta}{\lambda}$ has a compact support when $\beta > 1 + \frac{ \lambda - 1 }{ p }.$ In the case $\beta \leqslant 1 + \frac{ \lambda - 1 }{ p }$ the support is the whole real set and the density has an exponential decay for $\beta = 1 + \frac{ \lambda - 1 }{p}$ and a power-law decay for $\beta < 1 + \frac{ \lambda - 1 }{p}.$ On the other hand, when $1 + \beta - \lambda > 0$ and $\lambda > 1,$ the density $\gtd{p}{\beta}{\lambda}$ has a compact support when $\beta > 1$, and the real set for $\beta \leqslant 1$ and the density has an exponential decay for $\beta = 1$ and a power-law decay for $\beta < 1.$ 
\end{remark}
\begin{remark}\label{remark_bounded_2}
From remark~\ref{remark_bounded_1}, Corollary~\ref{coroll} and the integrability condition $\sign(1 + \beta - \lambda) = \sign\left( \beta + \frac{ 1 - \lambda }{p} \right)$, it follows that for $1 + \beta - \lambda > 0$ and $\beta > \frac{ \lambda -1 }{p}$,  i.e., for $\beta > \max\left(\lambda - 1,\frac{ \lambda - 1 }{p}\right)$, the density $\minpdf{p}{\beta}{\lambda}$ is bounded. For $1 + \beta - \lambda < 0$ and $\beta < \frac{ \lambda -1 }{p}$, i.e., for $\beta < \min\left(\lambda - 1,\frac{ \lambda - 1 }{p}\right)$, density $\minpdf{p}{\beta}{\lambda}$ is divergent in the borders of its compact support. In the sequel of the paper, we will mainly focus on the subfamily of bounded probability density functions. 
\end{remark}

As proved in~\cite{Zozor17}, the minimizers of the Stam inequality Eq.~\eqref{ineq:trip_Stam} satisfy the relation $\minpdf{p}{\beta}{\lambda} =\descort[\overline\alpha^{\, -1}]{\gauss_{p,\lambda_0}}$, and thus, for $\lambda \neq 1$, and $\beta \neq \lambda$, the minimizing densities $\minpdf{p}{\beta}{\lambda}$ are  generalized trigonometric densities, as  highlighted  when discussing equation~\eqref{ineq:trip_Stam}. In summary, the  family of minimizing densities of the generalized Stam inequality~\eqref{ineq:trip_Stam} (up to a scaling operation) can be expressed as\footnote{Although the densities given in Eq.~\eqref{def:rho_pbetalambda} can be defined for any parameters $p, \beta, \lambda$ such that $\sign(1-\lambda+\beta) = \sign(\frac{1-\lambda}p+\beta)$ (see definition~\eqref{def:GTDs} for GTDs and~\cite{Zozor17} for $\mathfrak G_{p,\beta}$), the generalized Stam inequality~\eqref{ineq:trip_Stam} only holds when $\lambda > 1 - \beta p^*$.}
\begin{equation}\label{def:rho_pbetalambda}
\minpdf{p}{\beta}{\lambda}(x) = \left\{
\begin{array}{ll}
\gauss_{p,\lambda}(x) & \mbox{if} \:\: \beta=\lambda,\\[2mm]
\mathfrak G_{p,\beta}(x) &  \mbox{if} \:\: \lambda=1,\\[2mm]
\gtd{p}{\beta}{\lambda}(x) & \mbox{otherwise}
\end{array}\right.
\end{equation}
where $\gauss_{p,\lambda}$ is given Eq.~\eqref{def:g_plambda}, $\gtd{p}{\beta}{\lambda}$ is given Eq.~\eqref{def:GTDs} and where
\begin{equation*}
\mathfrak G_{p,\beta}(x) = a_{p,\beta,1} \, \exp\left(-\frac{\Gamma^{-1}\left(\frac1{p^*},\,\frac{p^*}{(\beta^*)^{p^*}}|x|\right) }{\beta-1} \right),\qquad \beta\neq1,
\end{equation*} 
with $\Gamma^{-1}$  the inverse function  of the incomplete Gamma function, $a_{p,\beta,1}$ given in equation~\eqref{def:a_pbetalambda}, and the support $\support_{p,\beta,1}$ given by 
\[
\support_{p,\beta,1} = \left\{\begin{array}{ccc}
\left( - \, \frac{\left[ \Gamma\left( \frac1{p^*} \right) \right]^{1-\frac2\beta} \beta^{\frac1{p^*}}}{\left[ 4 p^* \right]^{\frac1\beta} (1-\beta)^{\frac1{p^*}}} \: , \: \frac{\left[ \Gamma\left( \frac1{p^*} \right) \right]^{1-\frac2\beta} \beta^{\frac1{p^*}}}{\left[ 4 p^* \right]^{\frac1\beta} (1-\beta)^{\frac1{p^*}}} \right)
& \mbox{if} & \beta  < 1, \\[7.5mm]
\Rset & \mbox{if} & \beta > 1.
\end{array}\right.
\]

These densities are widely studied in~\cite{Zozor17}, where the following relationship between the minimizers is also provided,
\begin{equation}\label{eq:trans_minpdf}
\descort[\alpha]{\minpdf{p}{\beta}{\lambda}} = \minpdf{p}{\frac\beta\alpha}{1+\frac{\lambda-1}\alpha}
\end{equation}

On the other hand, it is worth mentioning the existence of a symmetry-like relation between the minimizing densities summarized in the following Lemma.
\begin{lemma}[Zozor et al.~\cite{Zozor17}]\label{Lemma:rho_pbetalambda-symm_rel}
Let be $\minpdf {p}{\beta}{\lambda}$ a minimizer of the extended Stam inequality~\eqref{ineq:trip_Stam}. Then 
\begin{equation}\label{eq:rho_pbetalambda-symm_rel}
\minpdf {p}{\frac{\beta p^*+\lambda-1}{\lambda p ^*}}{\frac1\lambda} \propto \left[ \minpdf {p}{\beta}{\lambda} \right]^\lambda.
\end{equation}
\end{lemma}
Note that, for $\lambda \neq \beta,$ and $\lambda \neq 1$, the above Lemma~\ref{Lemma:rho_pbetalambda-symm_rel} can be easily derived through the Miyakawa duality relation~\eqref{eq:duality2}. Also, in the case $\lambda=\beta$, and after a re-parametrization $\lambda\mapsto\frac 1\lambda,$ in the above relation, one immediately obtains again the generalized Gaussian densities \ $\minpdf{p}{\frac{p^*+1-\lambda}{p^*}}{\lambda} = \gauss_{p,2-\lambda}.$

Note that the hyperbolic secant probability density belongs to the family of GTDs, obtained from the set of parameters $p = 2, \beta = 0$ and $\lambda = -1$, which arises performing the differential-escort transform $\descort[\frac12]{}$ of the Cauchy density $\gauss_{2,0}$. The probability density of the logistic distribution~\cite{JohKot95:v2}  also belongs to GTDs family from the set of parameters $p=2, \beta = \frac12$ and $\lambda = 0,$ which corresponds to the differential-escort density $\descort[\frac32]{\gauss_{2,2}}$. Finally, the probability density of the raised cosine distribution~\cite{Ahsanullah18} is also included in the family of generalized trigonometric densities, from $p=2, \, \beta = \frac32, \, \lambda=2$, or equivalently, from the differential-escort density $\descort[2]{g_{2,3}}$.


\subsection{Cumulative moments of $\minpdf{p}{\beta}{\lambda}$}

The objective of this section is to carry out a study of the cumulative moments of the whole family of densities $\minpdf{p}{\beta}{\lambda}$ involving GTDs, but also of that of stretched Gaussian, as well as that of the Gaussian distribution.


\subsubsection{Region of definition of the cumulative moments}

In the following lemma, in the conditions for which  $\minpdf{p}{\beta}{\lambda}$ is a bounded probability density, i.e., for $p^* > 0 $ and $\beta > \max\left( \lambda - 1, \frac{ \lambda - 1 }{ p } \right)$ (see remark~\ref{remark_bounded_2}), we give necessary and sufficient conditions of the cumulative moment parameters so that the corresponding cumulative moment is finite.
\begin{lemma}
Given  $p^*>0, \; \lambda, \beta$ such that $\beta > \max\left( \lambda - 1,\frac{ \lambda - 1}{p}\right)$, there exists a critical value
\[
\gamma_c = \beta + \frac{(1-\lambda)_+}p,
\]
such that, for any $q > 0$ and $\gamma >\gamma_c$, then
\begin{equation*}
\mu_{q,\gamma}\left[ \minpdf{p}{\beta}{\lambda} \right] \, < \, \infty \quad \Leftrightarrow \quad q \, < \, \frac{p\beta + (1-\lambda)_+}{p (\gamma-\beta) - (1-\lambda)_+}.
\end{equation*}
Moreover, 
$\mu_{q,\gamma}\left[ \minpdf{p}{\beta}{\lambda} \right]<\infty$ for all $q>0,$ and for any $\gamma\leqslant \gamma_c.$
\end{lemma}
\begin{proof}
First,  consider the case $\lambda_0 < 1$ in which
\[
\gauss_{p,\lambda_0}(x) \, \sim \, |x|^{-\frac{p^*}{1-\lambda_0}},\quad |x| \gg 1.
\]
meaning that the integrability requires $p^* > 1 - \lambda_0$.
Then, using Lemma~\ref{Lemma:tail}, when $\alpha > 1 - \frac{1-\lambda_0}{p^*}$, asymptotics $\descort[\alpha]{\gauss_{p,\lambda_0}}(y) \sim y^{-\frac{\alpha p^*}{\alpha p^* - p^* - \lambda_0} + 1}, \quad  |y| \gg 1$ is fulfilled. Note that, from integrability condition $1 - \frac{1-\lambda_0}{p^*} > 0$ parameter $\alpha$ is always positive. Now, using the parameters defined in Corollary~\ref{coroll} (i.e., replacing $\alpha$ by $\overline{\alpha}^{\, -1} = \frac1{1+\beta \lambda}$ and with $\lambda_0 = \frac{\beta}{1+\beta-\lambda}$), and after some algebraic manipulations one obtains: 
\begin{equation}\label{def:eta}
\minpdf{p}{\beta}{\lambda}(x) \, \sim \, x^{-\eta_{p,\beta,\lambda}}, \quad |x| \gg 1 \qquad \mbox{with \: } \eta_{p,\beta,\lambda}=\frac{p^*}{p^*(\lambda-\beta)+1-\lambda},
\end{equation} 
provided that
\begin{equation}\label{eq:cond_1}
p^* ( \lambda - \beta ) + 1 - \lambda > 0.
\end{equation}
Additionally, note that  from $\overline \alpha=1+\beta-\lambda>0$, the integrability condition of $\minpdf{p}{\beta}{\lambda}$ turns to be $\beta  > \frac{\lambda - 1}{p}$ (see Def.~\ref{definition:GTDs}), which implies that $\eta_{p,\beta,\lambda} > 1,$  with $\eta_{p,\beta,\lambda} \to 1$ when $\beta p + 1 - \lambda \to 0^+$, as expected. This set of remarks allows us to conclude that
\begin{equation*}
\mu_q\left[ \minpdf{p}{\beta}{\lambda} \right] \, < \,  \infty \quad %
\Leftrightarrow \quad \eta_{p,\beta,\lambda} - q > 1 \quad
\Leftrightarrow \quad q < \frac{1 + p \beta - \lambda}{\lambda + p \, (1-\beta) - 1}.
\end{equation*}
On the other hand, using the property~\ref{prop:mu_p-esc} we have $\mu_{q,\gamma}[\pdf] = \mu_q\left[ \descort[\gamma]{\pdf} \right],$  so
\begin{equation}\label{eq:muqgamma(g)}
\mu_{q,\gamma}\left[ \minpdf{p}{\beta}{\lambda} \right]
= \mu_q\left[ \descort[\gamma]{\minpdf{p}{\beta}{\lambda}} \right]
= \mu_q\left[ \minpdf{p}{\frac\beta\gamma}{\frac{\lambda-1}{\gamma}+1} \right],\qquad \gamma\neq 0.
\end{equation}
Now, by replacing $\beta$ by $\frac{\beta}{\gamma}$ and $\lambda$ by $\frac{\lambda - 1}{\gamma} + 1$ in Eq.~\eqref{eq:cond_1} and taking into account that $\beta + \frac{ 1 - \lambda }{p} > 0$, it follows from Eq.~\eqref{eq:muqgamma(g)} that, given $\gamma > \beta + \frac{ 1 - \lambda }{ p }$,
\begin{equation}\label{eq:eta(gamma)1}
\mu_{q,\gamma}\left[ \minpdf{p}{\beta}{\lambda} \right] \, < \, \infty \quad \Leftrightarrow \quad q < \eta_{p,\frac\beta\gamma,\frac{\lambda-1}{\gamma}+1}-1 = \frac{p \, \beta+1-\lambda}{p \, (\gamma-\beta)+\lambda-1}.
\end{equation}
Now, from Lemma~\ref{Lemma:mu_plambda-finite} and Eq.~\eqref{def:eta}, after simple algebraic manipulations one has that, for any $\gamma \leqslant \gamma_c$ the cumulative moments are finite, $\mu_{q,\gamma}[\minpdf{p}{\beta}{\lambda}]<\infty,$ where
\begin{equation}\label{eq:gamma_crit}
\gamma_c = \frac{1}{\eta_{p,\beta,\lambda}^*} = \beta+\frac{1-\lambda}p.
\end{equation}
Note that $\gamma_c$ is the value of $\gamma$ that vanishes the denominator of the right handside of the second equation of Eq.~\eqref{eq:eta(gamma)1}.

To analyze the case $\lambda > 1$, we consider the symmetry relation~\eqref{eq:rho_pbetalambda-symm_rel} in equation~\eqref{def:eta} to obtain
\begin{equation*}
\eta_{p,\beta,\lambda} = \frac1{1-\beta}, \quad \lambda > 1, \quad \beta \neq 1,
\end{equation*}
where, in contrast to what happens when $\lambda>1$, both parameters $p$ and $\lambda$ have disappeared.  Then,
\begin{equation}\label{eq:eta(gamma)2}
\eta_{p,\frac\beta\gamma,\frac{\lambda-1}{\gamma}+1} = \frac{\gamma}{\gamma-\beta},
\end{equation}
that leads to
$$
\mu_{q,\gamma}\left[\minpdf{p}{\beta}{\lambda} \right] \, < \, \infty \quad \Leftrightarrow \quad q < \frac{\gamma}{\gamma-\beta}-1 = \frac{\beta}{\gamma-\beta},$$
and  the critical value $\gamma_c$ is given by
\begin{equation}\label{eq:gamma_crit_lambda>1}
\gamma_c = \beta, \quad \lambda > 1.
\end{equation}

The case $\lambda = 1$ is more complicated to analyze in a direct way. However, both equations~\eqref{eq:eta(gamma)1} and~\eqref{eq:eta(gamma)2} coincide in the limit $\lambda \to 1.$ Hence, from continuity arguments, one has that in the case $\lambda = 1$ equations~\eqref{eq:eta(gamma)1},~\eqref{eq:gamma_crit},~\eqref{eq:eta(gamma)2} and~\eqref{eq:gamma_crit_lambda>1} remain valid.
\end{proof}

Let us now pay a special attention to the densities with exponential decay. The stretched Gaussian densities are recovered when $\lambda = \beta = 1$, i.e., $\gauss_{p,1}(x) = e^{-|x|^{p^*}}$.

In the case  $\lambda < 1$, the critical condition $\beta=1+\frac{\lambda-1}p$ (see remark~\ref{remark_support_length_2}), and the definitions~\ref{def:GTDs} and~\ref{def:rho_pbetalambda} imply that the densities are of hyperbolic cosine type, $\minpdf{p}{\beta}{\lambda }\propto \cosh_{p^*,p^*}\left( \kappa_{p,\beta , \lambda} \, y \right)^{\frac{p ^*}{\lambda-1}} = \cosh_{p^*}\left( \kappa_{p,\beta , \lambda} \, y \right)^{\frac{p ^*}{\lambda-1}},$ which have an exponential decay.

In the case $\lambda > 1$, we obtain the symmetrical case when $\beta = 1$, for which $\minpdf{p}{1}{\lambda}$ is given through the generalized cosine function $\cos_{1,p^*}$.  By definition\footnote{The generalized $v-$tangent function is defined as $ \tan_v = \frac{\sin_v(x)}{\cos_v(x)}$ where the notation $\sin_v = \sin_{v,v}$ and $\cos_v = \cos_{v,v}$ is usually adopted in the literature~\cite{Yin19}.}
\begin{equation}\label{def:arctan_p}
\arsinh_{1,p}(x) \, = \, \int_0^x \frac{dt}{1+|t|^p} \, \equiv \, \arctan_p(x)
\end{equation}
so that $\sinh_{1,p}(x)=\tan_p(x)$ and $\cosh_{1,p}(x)=\frac1{\cos_p^p(x)}, $ precisely corresponding to the symmetrical case of the above cosine densities.

The above-mentioned densities are the only ones in the family $\minpdf{p}{\beta}{\lambda}$ with an exponential decay.

\begin{remark}
From Eq.~\eqref{eq:gplalpha}, with $\lambda<1$, and imposing $\frac{1-\lambda}{1-\alpha}=p^*$ (or equivalently $\alpha=1+\frac{\lambda-1}{p^*} $)  one has that $\descort[\alpha]{\gauss_{p,\lambda}}$ has an exponential decay and then each moment is finite. Then, $m_i=\mu_{\alpha,i}[\gauss_{p,\lambda}]$ is finite for any $i\in\mathbb N,$ as far as they correspond to the moments of an exponential decaying GTD with the form $\cosh_{p^*}\left( \kappa_{p,\beta , \lambda} \, y \right)^{\frac{p ^*}{\lambda-1}}.$ In the same line as in Section~\ref{ssec:cauchy}, stretched Gaussian densities are characterized by $\left\{ \left\{ 1 +\frac{\lambda-1}{p ^*} \right\} \, , \, \{m_i\}_{i=1}^{+\infty} \right\}.$ In this sense, this subfamily of GTDs can be understood as an auxiliary family of functions with all their moments finite, which allows to recover the whole stretched Gaussian densities family.
\end{remark}


\section{Cumulative inequalities}\label{Sec_CumIneq}

In this last section, we show that generalized cumulative moments are lower bounded by Rényi and Shannon entropies, as well as, by the inverse of the family of biparametric generalized Fisher informations. In addition, the lower bounds are achieved for the same family of densities $\minpdf{p}{\beta}{\lambda}$ which minimizes the extension of Stam inequality with three parameters proposed in~\cite{Zozor17}. Here, the third parameter is the one that appears in the generalized cumulative moments. The study is restricted to the parameters such that the minimizing densities are bounded.


\subsection{Entropy-momentum like inequality}

In the following lines we show that, for a fixed generalized cumulative moment, the family of densities $\minpdf{p}{\beta}{\lambda}$ maximizes the Rényi entropies (and the Shannon one as a particular case), thus generalizing the moment-entropy relations of~\cite{Lutwak04, Bercher12}. In particular, GTDs and $\mathfrak G_{p,\beta}$ maximize Rényi and Shannon entropies respectively.
\begin{theorem}\label{EMineq3}
Let be a real number\footnote{The case $p^* \to +\infty$ given by $p = 1$ is also included in the result.} $p^* \in [0,\infty)$, and two real numbers $\lambda$ and $\beta$ such that $ \beta > \max\left( \lambda-1 , \frac{1-\lambda} p \right).$ Then, for any continuously differentiable probability density $\pdf$, the following inequality is fulfilled
\begin{equation}\label{cumulative_EM}
\frac{\sigma_{p^*,1+\beta-\lambda}[\pdf]}{N_\lambda[\pdf]} \, \geqslant \, \frac{\sigma_{p^*,1+\beta-\lambda}[\minpdf {p}{\beta}{\lambda}]}{N_\lambda[\minpdf {p}{\beta}{\lambda}]} \, \equiv \, \overline{K}_{p,\beta,\lambda}.
\end{equation} 
Moreover, the family of minimizing densities is given by $\minpdf{p}{\beta}{\lambda},$ up to a scaling transformation. The optimal bound is given by
\begin{equation*}
\overline{K}_{p,\beta,\lambda} = \left\{\begin{array}{lll}
\displaystyle \frac{ \left(p^*\beta\right)^{\frac1{\lambda-1}} \, (1+\beta-\lambda)^{\frac{1}{p^*(1+\beta-\lambda)}} }{\left(p^*\beta+\lambda-1\right)^{\frac{1}{\lambda-1}+\frac1{p^*(1+\beta-\lambda)}} } \: a_{p,\beta,\lambda}^\frac1{1+\beta-\lambda} & \mbox{for}  & \lambda \neq 1,\\[7.5mm]
\displaystyle  \left(\frac{(p^*)^{\frac1{p}}\,e^{-\frac1{p^*}}}{2\,\Gamma(\frac1{p^*})}\right)^{\frac 1{\beta}} & \mbox{for} &  \lambda = 1.
\end{array}
\right.
\end{equation*}
\end{theorem}
\begin{proof}
The detail of the proof, postponed to Appendix~\ref{App_proofs}, mainly lies on the application of differential-escort transforms to inequality~\eqref{ineq:bip_E-M}.
The key points are the following. Starting from the inequality~\eqref{ineq:bip_E-M}, and using the composition property~\ref{prop:composition} of the differential-escort transformation one finds that, for any $p^* \in [0,\infty), \; \lambda > \frac1{1+p^*},$ and $\alpha>0,$
\begin{equation*}
\frac{\sigma_{p^*}\left[ \descort[\alpha\alpha^{-1}]{\pdf} \right]}{N_\lambda\left[ \descort[\alpha\alpha^{-1}]{\pdf} \right] } = \frac{ \sigma_{p^*}[ \pdf ] }{ N_\lambda[ \pdf ] } \geqslant \overline{K}_{p,\lambda}.
\end{equation*}
Thus, using the stretching effect on Rényi entropies and cumulative moments given in  properties~\ref{prop:Renyi-esc} and~\ref{prop:mu_p-esc}, it follows that
\begin{equation*}
\left( \frac{ \sigma_{p^*,\alpha} \left[ \descort[\alpha^{-1}]{\pdf} \right]}{ N_{\lambda_\alpha} \left[ \descort[\alpha^{-1}]{\pdf} \right]} \right)^\alpha \geqslant \overline K_{p,\lambda}.
\end{equation*}
with $\lambda_\alpha = 1 + \alpha (\lambda - 1)$ given in Eq.~\eqref{Eq:lambda_alpha} of property~\ref{prop:Renyi-esc}.
Finally, taking into account that the set of differentiable densities $\D$ is closed for the differential-escort transform, $\descort[\alpha^{-1}]{\D} = \D$, as discussed in references~\cite{Puertas19, Zozor17}, we can rename the parameters to close the proof.
\end{proof}
As previously evoked, analogous inequalities can be obtained by substituting Rényi entropies by Tsallis entropies.

Note that, as an application of theorem~\ref{EMineq3}, the probability density of the raised cosine distribution can be characterized as being the maximal $\lambda-$Rényi entropy with $\lambda=2$ (directly related with the so-called disequilibrium $D[\pdf]=e^{-R_2[\pdf]}$) when $\sigma_{2,\frac12}[\pdf]$ is fixed. However, although the logistic and the hyperbolic secant probability densities are encompassed by GTDs, they fall outside the application of this theorem.


\subsection{Cramér-Rao like inequality}
The following theorem trivially follows by multiplying triparametric Stam like inequality~\eqref{ineq:trip_Stam}, and the entropy-momentum like stated in Theorem~\ref{EMineq3}:
\begin{theorem}\label{CRineq3}
Let be a real number\footnote{Again, in the limit $p \to +\infty,$ the inequality is still valid.} $p \in [1,\infty)$, and two real numbers $\lambda$ and $\beta$ such that $ \beta > \max\left(\lambda-1 , \frac{1-\lambda} p \right).$ Then, for any continuously differentiable probability density $\pdf$, the following inequality is fulfilled:
\begin{equation}\label{cumulative_CR}
\phi_{p,\beta}[\pdf] \, \sigma_{p^*,1+\beta-\lambda}[\pdf] \, \geqslant \, \phi_{p,\beta}[\minpdf{p}{\beta}{\lambda}] \, \sigma_{p^*,1+\beta-\lambda}[\minpdf {p}{\beta}{\lambda}] \, \equiv \, K_{p,\beta,\lambda}.
\end{equation} 
The family of minimizing densities is given by $\minpdf{p}{\beta}{\lambda}$ and the optimal bound is given by
\begin{equation*}
K_{p,\beta,\lambda} = (p^*)^{\frac 1\beta}\, \left( \frac{(1+\beta-\lambda)^{\frac{\lambda-1}{p^*}} \: a_{p,\beta,\lambda}^{\lambda-1}}{(p^*\beta+\lambda-1)^{\frac{p\beta+\lambda-1}{p}}}\right)^\frac{1}{\beta(1+\beta-\lambda)}.
\end{equation*}
\end{theorem}
\begin{proof}
The proof trivially follows from equations~\eqref{ineq:trip_Stam} and~\eqref{cumulative_EM}. To obtain the optimal bound we simply express $\sigma_{p^*}[\gauss_{p,\lambda_0}]$ and $\phi_{p,\lambda_0}[\gauss_{p,\lambda_0}]$ from~\cite{Lutwak05} and later we apply properties~\eqref{prop:Fisher-esc},~\eqref{prop:mu_plambda-esc} and Corollary~\eqref{coroll}.
\end{proof}
Again, the probability density function of the raised cosine distribution minimizes the generalized $(2,3/2)-$Fisher information when the $(2,1/2)-$th order cumulative moment is fixed.  Finally, note that the functionals $\phi_{p,\beta}[\cdot] \, \sigma_{p^*,1+\beta-\lambda} [\cdot]$ and $\frac{\sigma_{p^*,1+\beta-\lambda}[\cdot]}{N_{\lambda}[\cdot]}$ are invariant under scaling transformations, but not under translation transformations. This drawback can be solved by the introduction of a central version of the cumulative moments, as done in Appendix~\eqref{app:central} where it is also shown that similar inequalities are satisfied with such central cumulative moments.


\section{Conclusions}\label{Sec_Conclusions}

In this paper, we mainly focused on probability densities linked to generalized trigonometric functions, and on their roles in extended informational inequalities. First, we defined a three-parameters family of probability densities, named \textit{generalized trigonometric densities}, which include the Gaussian density, generalized Gaussian densities, as well as hyperbolic secant, logistic, and raised cosine densities. This was precisely done by means of the so-called generalized trigonometric functions introduced by Lindqvist~\cite{Lindqvist95} and Drabek~\cite{Drabek99}. Moreover, we prove that for a fixed generalized Fisher information, the Rényi entropies (and the Shannon one as a particular case) are minimized by the introduced family of densities.  In addition, a bi-parametric class of generalized moments called \textit{cumulative moment}, was introduced. These moments are defined as the mean of the power of a deformed cumulative distribution, which, in turn, is defined via the cumulative function of the power of the probability density function. This second parameter tunes the tail weight of the henced defined deformed cumulative distribution. The usual moments are recovered as particular cases of these generalized moments, and other remarkable properties like order relation w.r.t. the first parameter, or adequate scaling behavior are proven.  In particular, we show that, for any bounded probability density, there exists a critical value for the second parameter (power of the deformed cumulative function) below which the whole subfamily of generalized moments is finite for any positive value of the first parameter (power applied to the deformed cumulative). In this generalized framework we have extended the entropy-moment and Cramér-Rao inequalities, showing that for a fixed generalized moment, both the Rényi entropies and the generalized Fisher information achieve their maximum and minimum values, respectively, precisely for the generalized trigonometric densities. In addition, we have highlighted the fact that GTDs and cumulative moments allow to formally reconstruct the whole subfamily of stretched Gaussian densities.

A first direction of perspectives to this work is to extend the study to multivariate densities. This is not an easy task especially dealing with such multivariate extension of the generalized Fisher information and of the cumulative moments, but also of the definition of the differential-escort transform itself since there is no unique way to impose the probability conservation (stretching the differential interval of the cartesian axes in the same manner, stretching radially,\ldots). More especially, one has to think about how to extend the transform and these measures (generalized Fisher information, cumulative moments) so that they conserve the ``nice'' properties of the scalar context derived in this paper. Even more difficult seems to be the extension of these ideas to the matrix context in which the Fisher information matrix gives rise to richer inequalities than their restriction obtained from the trace of this matrix (i.e., involving a norm of the gradient).


\subsection*{Acknowledges}

DPC is very grateful to the  Ministry of Science and Innovation of Spain for its partial financial support under grants PID2020-115273GB-100 (AEI/FEDER, EU) and PID2023-153035NB-100 funded by \text{MICIU/AEI/10.13039/501100011033} and ERDF/EU
, as well as, for the warm hospitality during his stay at GIPSA–Lab of the University of Grenoble–Alpes where this work was partially carried out.

\subsection*{Conflict of interest statement}

The authors declare no conflicts of interest.

\subsection*{Data availability statement}

Data sharing is not applicable to this article as it has no associated data.


\newpage \appendix 


\section{Proof of Lemma~\ref{Lemma:GTDs&gplambda}} \label{App1}

Given real numbers $p$ and $\lambda$ such that $p^* > 0$ and $\lambda > 1 - p^*, \, \lambda\neq 1$, the generalized  Gaussian density $g_{p,\lambda}: \support_{p,\lambda} \longrightarrow \Rset$, and a real $\alpha \neq 1$, the change of variables $y=y(x)$  given in definition~\ref{definition:diff-esc} can be computed as follows.


\subsection*{Transformed density expression}

First, for $\lambda < 1$ and $x \in \support_{p,\lambda},$
\begin{eqnarray*}
y(x) & = & a_{p,\lambda}^{1-\alpha}  \int_0^x \left( 1 + \left( (1 - \lambda)^\frac1{p^*} |t| \right)^{p^*} \right)_+^{\frac{1-\alpha}{\lambda-1}} \, dt\\[2.5mm]
& = & \frac{a_{p,\lambda}^{1-\alpha} }{ (1-\lambda)^\frac1{p^*} } \int_0^{(1-\lambda)^\frac1{p^*}\,x} (1 + |u|^{p^*})_+^{ \frac{1-\alpha}{\lambda-1} } \, du\\[2.5mm]
& = & \frac{a_{p,\lambda}^{1-\alpha}}{(1-\lambda)^\frac{1}{p^*}} \, \arsinh_{\frac{1-\lambda}{1-\alpha},p^*}\left( (1-\lambda)^\frac1{p^*} x \right).
\end{eqnarray*}
from definition in equation~\eqref{def:arcsinhpq} defining the generalized arsinh function.

In a similar way,  from equation~\eqref{def:arcsinpq} defining the generalized arcsin function, for $\lambda > 1$ one obtains
\begin{equation*}
y(x) = \frac{a_{p,\lambda}^{1-\alpha}}{(\lambda-1)^\frac1{p^*}} \, \arcsin_{\frac{1-\lambda}{1-\alpha},p^*}\left( (\lambda-1)^\frac1{p^*} x \right).
\end{equation*}

From these expressions, after simple algebraic manipulations, one gets 
\begin{equation*}
x(y) = \left\{\begin{array}{ccc}
|1-\lambda|^{-\frac1{p^*}} \sinh_{\frac{1-\lambda}{1-\alpha},p^*}\left( \widetilde y \right)& \mbox{if} & \lambda < 1,\\[2.5mm]
|1-\lambda|^{-\frac1{p^*}} \sin_{\frac{1-\lambda}{1-\alpha},p^*}\left( \widetilde y \right)& \mbox{if} & \lambda > 1
\end{array}\right. 
\qquad \mbox{with} \quad \widetilde y = \frac{|1-\lambda|^\frac{1}{p^*}}{a_{p,\lambda}^{1-\alpha}} \, y
\end{equation*}

Now, we apply  the definition~\ref{definition:diff-esc}, $\descort[\alpha]{g_{p,\lambda}}(y) = \left[ g_{p,\lambda}(x(y)) \right]^\alpha$ so that, using the pytheagorean like relation~\eqref{eq:pyth_rel_hyperbolic}, for $\lambda < 1$  it follows
\begin{eqnarray*}
\descort[\alpha]{g_{p,\lambda}}(y)
& = & a_{p,\lambda}^\alpha \left( 1 + (1-\lambda) \left[ (1-\lambda)^{-\frac1{p^*}} \, \sinh_{\frac{1-\lambda}{1-\alpha},p^*} \left( \widetilde y \right) \right]^{p^*} \right)_+^\frac{\alpha}{\lambda-1}\\[2.5mm]
& = & a_{p,\lambda}^\alpha \left( 1 + \left[ \sinh_{\frac{1-\lambda}{1-\alpha},p^*} \left( \widetilde y \right) \right]^{p^*} \right)_+^\frac{\alpha}{\lambda-1}\\[2.5mm]
& = & a_{p,\lambda}^\alpha \left( \left[ \cosh_{\frac{1-\lambda}{1-\alpha},p^*} \left( \widetilde y \right) \right]^{\frac{1-\lambda}{1-\alpha}} \right)_+^\frac{\alpha}{\lambda-1}\\
& = & a_{p,\lambda}^\alpha \left( \cosh_{\frac{1-\lambda}{1-\alpha},p^*} \left( \widetilde y \right) \right)_+^\frac{\alpha}{\alpha-1}.
\end{eqnarray*}
The case $\lambda > 1$ is treated similarly, leading to 
\begin{equation*}
\descort[\alpha]{g_{p,\lambda}}(y) = a_{p,\lambda}^\alpha \left( \cos_{\frac{1-\lambda}{1-\alpha},p^*} \left( \widetilde y \right) \right)_+^\frac{\alpha}{\alpha-1},
\end{equation*}
where the pytheagorean like relation~\eqref{eq:pyth_rel} has been used.


\subsection*{Transformed support $(\support_{p,\lambda})_\alpha$}


The transformed support $(\support_{p,\lambda})_\alpha = y\left( \support_{p,\lambda} \right) = \left( -  y_{p,\lambda,\alpha}  \, , \,  y_{p,\lambda,\alpha}  \right)$ can be obtained by applying the change of variables to the border points of the support  $\support_{p,\lambda} = ( - x_{p,\lambda} \, , \, x_{p,\lambda} ),$ with $x_{p,\lambda} = \left( \frac1{\lambda-1} \right)^{\frac1p}$ when $\lambda > 1$ and  $x_{p,\lambda} = +\infty$ for $\lambda < 1.$ For $\lambda < 1$ one has:
\[
y_{p,\lambda,\alpha} \equiv y\left( x_{p,\lambda} \right) =  \frac{ a_{p,\lambda}^{1-\alpha }}{ (1-\lambda)^\frac1{p^*} } \int_0^\infty \left( 1 + u^{p^*} \right)^{ \frac{1-\alpha}{\lambda-1} } \, du
\]
which is finite if and only if $p^* \left( \frac{ 1-\alpha }{ \lambda-1 } \right) < -1.$ Since $p^*>0$, $y_{p,\lambda,\alpha}$ is then finite when 
\[
\alpha < 1 - \frac{1-\lambda}{p^*}
\]
and in such  case  
\[
y_{p,\lambda,\alpha} = \frac{ a_{p,\lambda}^{1-\alpha} }{ p^*\,(1-\lambda)^\frac{1}{p^*} } \int_0^\infty ( 1 + v )^{ \frac{ 1-\alpha }{ \lambda-1 } } \, v^{-\frac1p} \, dv = 
\frac{ a_{p,\lambda}^{1-\alpha} }{ p^* \, ( 1-\lambda )^\frac{1}{p^*} } B\left( \frac1{p^*}, \; \frac{ 1-\alpha }{ 1-\lambda } - \frac1{p^*} \right)
\]
from~\cite[Eq.~3.194]{Gradshteyn15}.
For $\alpha \geqslant 1-\frac{1-\lambda}{p^*}, \; y_{p,\lambda,\alpha} = +\infty$. The quantity  $y_{p,\lambda,\alpha}$ can be expressed for any $\alpha$ in a more compact way using the generalized \textit{pi} constant given in~Eq.~\eqref{def:pipq},
\[
y_{p,\lambda,\alpha} = \frac{ a_{p,\lambda}^{1-\alpha} \, \pi_{\chi,p^*} }{ 2\,(1-\lambda)^\frac{1}{p^*}}  
\qquad \mbox{with} \qquad \chi = \left(\frac{p^*+1}{p^*}-\frac{1-\alpha}{1-\lambda}\right)^{-1} = \left( \frac{p^*+1}{p^*} \one_{\Rset_+} (1-\lambda) + \frac{\alpha-1}{|1-\lambda|} \right)^{-1}.
\]
We remark that $\pi_{\chi,p^*}<\infty$ for $\chi\in\Rset /[0,1],$ and infinite in other case.


In the case $\lambda>1$ we have that 

\[
y_{p,\lambda,\alpha} = \frac{ a_{p,\lambda}^{1-\alpha} }{ (\lambda-1)^\frac1{p^*} } \int_0^1 \left( 1-u^{p^*} \right)_+^{ \frac{ 1-\alpha }{ \lambda-1 } } \, du.
\]
This quantity is finite if and only if $\frac{ 1-\alpha }{ \lambda-1 } > -1$,
which is equivalent to 
\[
\alpha < \lambda.
\]
and in this case 
\[
y_{p,\lambda,\alpha} = \frac{ a_{p,\lambda}^{1-\alpha} }{ (\lambda-1)^\frac1{p^*} } \, \arcsin_{\frac{1-\lambda}{1-\alpha},p^*} (1) = \frac{  a_{p,\lambda}^{1-\alpha} \, \pi_{\xi,p^*} }{ 2 \, (\lambda-1)^\frac{1}{p^*} }
\qquad \mbox{with} \qquad \xi = \frac{ \lambda-1 }{ \alpha-1 } = \left( \frac{p^*+1}{p^*} \one_{\Rset_+} (1-\lambda) + \frac{\alpha-1}{|1-\lambda|} \right)^{-1}
\]
Again, we remark that $\pi_{\xi,p^*} < \infty$ for $\xi \in \Rset / [0,1],$ and infinite in other case.


\section{Proof of Theorem~\ref{EMineq3}}\label{App_proofs}

Let be $\pdf$ a probability density continuously differentiable, and a positive number $\alpha>0.$ Using the composition property~\eqref{prop:composition} one gets $\descort[\alpha\alpha^{-1}]{\pdf} = \pdf.$ In addition, by virtue of the properties~\ref{prop:Renyi-esc} and~\ref{prop:mu_p-esc},
\begin{equation*}
\frac{\sigma_{p^*}[\pdf]}{N_\lambda[\pdf]} = \left( \frac{\sigma_{p^*,\alpha}\left[ \descort[\alpha^{-1}]{\pdf} \right]}{N_{\lambda_\alpha}\left[ \descort[\alpha^{-1}]{\pdf} \right]} \right)^\alpha.
\end{equation*}
with $\lambda_\alpha = 1 + \alpha (\lambda-1)$ given eq.~\eqref{Eq:lambda_alpha} of property~\ref{prop:Renyi-esc}. Now, for  $\pdf$ that satisfies  the conditions of the generalized entropy-momentum inequality (see Eq.~\eqref{ineq:bip_E-M}, and references~\cite{Lutwak04, Bercher12}), and  real numbers $p, \; \lambda$ such that $\lambda > \frac1{1+p^*},\; p^*\in[0,\infty),$ the following inequality  is fulfilled
\begin{equation}\label{ineq:1}
\left( \frac{\sigma_{p^*,\alpha}\left[ \descort[\alpha^{-1}]{\pdf} \right]}{N_{\lambda_\alpha}\left[ \descort[\alpha^{-1}]{\pdf} \right]} \right)^\alpha
\geqslant
\frac{\sigma_{p^*}[g_{p,\lambda}]}{N_\lambda[g_{p,\lambda}]},
\end{equation}
with  equality for $\pdf = g_{p,\lambda}.$  Since $\descort[\alpha^{-1}]{\D} = \D$ being $\D$ the set of continuously differentiable probability densities, by a change of notation the last equation trivially rewrites as
\begin{equation}\label{ineq:2}
\left(\frac{\sigma_{p^*,\alpha}[\pdf]}{N_{\lambda_\alpha}[\pdf]}\right)^\alpha
\geqslant
\frac{\sigma_{p^*}[g_{p,\lambda}]}{N_{\lambda}[g_{p,\lambda}]}.
\end{equation}
As the minimum of the left handside of inequality~\eqref{ineq:1} is reached for $\pdf = g_{p,\lambda}$, the minimum of the left handside of inequality~\eqref{ineq:2} is reached  for $\pdf = \descort[\alpha^{-1}]{g_{p,\lambda}}$. Written for the parameters $\overline \alpha = 1 + \beta - \lambda$ instead of $\alpha$, and $\lambda_0 = \frac{\beta}{1 + \beta - \lambda} > 0$ instead of $\lambda$, the inequality can be then rewritten, for $\lambda_0 > \frac1{1+p^*}$ and $p^* \in [0,\infty)$,  under the form
\begin{equation*}
\frac{\sigma_{p^*,1+\beta-\lambda}\left[ \pdf \right]}{N_{\lambda}\left[ \pdf \right]}
\geqslant
\left(\frac{\sigma_{p^*}\left[ g_{p,\lambda_0} \right]}{N_{\lambda_0}\left[ g_{p,\lambda_0} \right]}\right)^\frac1{1+\beta-\lambda} \equiv \overline K_{p,\beta,\lambda},
\end{equation*}
%
with equality for $\pdf = \descort[\overline \alpha^{\, -1}]{g_{p,\lambda_0}}\equiv \minpdf{p}{\beta}{\lambda}$ , up to a scaling, from Corollary~\eqref{coroll}.
Note then that the condition $p^*\in[0,\infty)$ remains invariant. The condition $\overline \alpha>0,$ leads to $\beta>\lambda-1.$ On the other hand, the condition $\lambda_0>\frac1{1+p^*},$ joined with the previous one, impose that $\beta>\frac{1-\lambda}{p^*}$ after simple algebraic manipulations. Note that both conditions can be summarized  as $\beta > \max\left( \lambda-1 \, , \, \frac{1-\lambda}{p^*} \right)$.

Finally, the minimal bound can be easily computed, noting that
\[
\overline K_{p,\beta,\lambda} = \frac{\sigma_{p^*,1+\beta-\lambda}\left[ \minpdf{p}{\beta}{\lambda} \right]}{N_\lambda\left[ \minpdf{p}{\beta}{\lambda} \right]} = \left( \frac{\sigma_{p^*}\left[ g_{p,\lambda_0} \right]}{N_{\lambda_0}\left[ g_{p,\lambda_0} \right]} \right)^\frac1{\overline\alpha}.
\]
Then, the power Rényi entropy is computed as~\cite{Lutwak05}
\begin{equation*}
N_{\lambda_0}\left[ \gauss_{p,\lambda_0} \right] = \left\{\begin{array}{lll}
\left( \frac{p^* \lambda_0}{p^* \lambda_0 + \lambda_0 - 1} \right)^\frac1{1-\lambda_0} \, a_{p,\lambda_0}^{-1} & \mbox{for}  & \lambda_0 \neq 1,\\[7.5mm]
\displaystyle e^\frac1{p^*}\frac{2 \, \Gamma(\frac1{p^*})}{p^*} & \mbox{for} &  \lambda_0 = 1.
\end{array}
\right.
\end{equation*}
Simple algebra allows to obtain for the $p^*-$order standard deviation the expression
\begin{equation*}
\sigma_{p^*}\left[ \gauss_{p,\lambda_0} \right] = \left( p^* \lambda_0 + \lambda_0 - 1 \right)^{-\frac1{p^*}}
\end{equation*}
The expression of the bound follows from the compilation of these two results and some algebraic manipulation.


\section{Cumulative central moments}\label{app:central}

A lack of the family of generalized cumulative moments is the fact that it is not shift-invariant (i.e., invariant under translation transformations), in contrast to what happens with the family of central moments and their generalization. The aim of this appendix is to propose a centered version of the cumulative moments which are precisely shift-invariant. In particular, these moments correspond with the central moment of the differential-escort transformed density. However, requiring the shift-invariance notably complexifies the definition of such moments compared to the non-centered case. 
\begin{definition}[Cumulative central moment] Given two real numbers $p>0$, $\lambda\in\Rset$, and a probability density $\pdf$. Then we denote by $\mu_{p,\lambda}^{(c)}[\pdf]$ the $(p,\lambda)-$cumulative central moment of $\pdf$,  and define it as
\begin{eqnarray*}
\mu_{p,\lambda}^{(c)}[\pdf]
& = & \int_\Rset \left| \int_{\Rset} \left( \int_u^x [\pdf(t)]^{1-\lambda} dt \right) \pdf(u) du \right|^p \pdf(x) \, dx\\[2mm]
&=& \left\langle \left| \wcdf{\lambda}(x) - \left\langle \wcdf{\lambda}(x) \right\rangle_\pdf \right|^p \right\rangle_\pdf \nonumber
\end{eqnarray*}
\end{definition}
The hence defined central cumulative moment $\mu_{p,\lambda}^{c}$ is invariant under translation transformations,
\begin{property}[Shift-invariance of the cumulative central moments]
Given the translation transform $\pdf_{[a]}(x) = \pdf(x-a)$ for any $a \in \Rset$, 
\begin{equation*}
\mu_{p,\lambda}^{(c)}\left[ \pdf_{[a]} \right] = \mu_{p,\lambda}^{(c)}[\pdf]
\end{equation*}
\end{property}
\begin{proof}
From the definition, by successive changes of variables $t \to t-a, \; u \to u-a \; , x \to x-a$, one has
\begin{eqnarray*}
\mu_{p,\lambda}^{(c)}[\pdf_{[a]}]
& = & \int_\Rset \left| \int_{\Rset} \left( \int_u^x [\pdf(t-a)]^{1-\lambda} \, dt \right) \, \pdf(u-a) \, du \right|^p \, \pdf(x-a) \, dx\\[2mm]
& = &\int_\Rset \left| \int_{\Rset} \left( \int_{u-a}^{x-a} [\pdf(t)]^{1-\lambda} \, dt \right) \, \pdf(u-a) \, du \right|^p \, \pdf(x-a) \, dx\\
&=&\int_{\Rset} \left| \int_\Rset \left( \int_u^{x-a} [\pdf(t)]^{1-\lambda} \, dt \right) \, \pdf(u) \, du \right|^p \, \pdf(x-a) \, dx\\[2mm]
& = & \int_\Rset \left| \int_\Rset \left( \int_u^x [\pdf(t)]^{1-\lambda} \, dt \right) \, \pdf(u) \, du \right|^p \, \pdf(x) \, dx\\[2mm]
& = & \mu_{p,\lambda}^{(c)}[\pdf].
\end{eqnarray*}
\end{proof}

In the following lines we show that these centered moments also satisfy entropy-momentum and Cramér-Rao like inequalities. First, note that the products $\sigma_p^{(c)}[\pdf]\phi_{p,\lambda}[\pdf]$ and $\frac{\sigma_p^{(c)}[\pdf]}{N{\lambda}[\pdf]}$ are invariant under translation operation of the probability density. Then, given any probability density $\pdf$ with a finite mean $\langle x\rangle_\pdf$, one can consider the zero-mean density $\pdf_{[\langle x \rangle_f]}$, for which $\sigma_p^{(c)} \left[ \pdf_{[\langle x \rangle_f]} \right] = \sigma_p\left[ \pdf_{[\langle x \rangle_f]} \right]$.
Consequently,
$$\sigma_p^{(c)}[\pdf] \phi_{p,\lambda}[\pdf] = \sigma_p^{(c)}\left[ \pdf_{[\langle x \rangle_f]} \right] \phi_{p,\lambda} \left[ \pdf_{[\langle x \rangle_f]} \right] = \sigma_p \left[ \pdf_{[\langle x \rangle_f]} \right] \phi_{p,\lambda} \left[ \pdf_{[\langle x \rangle_f]} \right] \, \geqslant \, K_{p,\lambda}.$$ 
Similarly,

$$\frac{ \sigma_p^{(c)}[\pdf] }{ N_{\lambda}[\pdf] } = \frac{ \sigma_p^{(c)} \left[ \pdf_{[\langle x \rangle_f]} \right] }{ N_{\lambda} \left[ \pdf_{[\langle x \rangle_f]} \right] } = \frac{ \sigma_p \left[ \pdf_{[\langle x \rangle_f]} \right] }{ N_{\lambda} \left[ \pdf_{[\langle x \rangle_f]} \right] } \, \geqslant \, \overline K_{p,\lambda}.$$ 
In conclusion, inequalities~\eqref{cumulative_CR} and~\eqref{cumulative_EM} are still valid when $\sigma_{p,\lambda}[\pdf]$ is substituted by $\sigma^{(c)}_{p,\lambda}[\pdf],$ with the same minimal bounds and minimizing densities (up to not only scaling changes but also translations). Moreover, all other properties of the generalized cumulative moments can be extended to the central version without any additional cost. 

\bibliographystyle{unsrt}
\bibliography{refs}
\end{document}